\def\comments{0}

\documentclass[11pt]{article}
\usepackage[utf8]{inputenc}
\usepackage{amsmath,amssymb,bbm}
\usepackage{preamble}
\usepackage{url}

\usepackage{authblk}

\newcommand{\indic}[1]{\mathbbm{1}\left[#1\right]}

\newcommand{\rr}{\texttt{RR}}
\newcommand{\rrr}{\ensuremath{R^\texttt{RR}}}
\newcommand{\vecrrr}{\ensuremath{\vec{R}^\texttt{RR}}}
\newcommand{\rra}{\ensuremath{A^\texttt{RR}}}

\newcommand{\hst}{\texttt{HST}}
\newcommand{\hstr}{\ensuremath{R^\texttt{HST}}}
\newcommand{\hsta}{\ensuremath{A^\texttt{HST}}}

\newcommand{\hh}{\texttt{HH}}
\newcommand{\hhr}{\ensuremath{R^\texttt{HH}}}
\newcommand{\hha}{\ensuremath{A^\texttt{HH}}}

\newcommand{\esti}{\ensuremath{{\texttt{EST}\infty}}}
\newcommand{\estir}{\ensuremath{R^\esti}}
\newcommand{\estia}{\ensuremath{A^\esti}}

\newcommand{\esto}{\texttt{EST1}}
\newcommand{\estor}{\ensuremath{R^\esto}}
\newcommand{\estoa}{\ensuremath{A^\esto}}

\newcommand{\estt}{\texttt{EST2}}

\newcommand{\raptor}{\texttt{RAPTOR}}
\newcommand{\raptorr}{\ensuremath{R^\texttt{RAPTOR}}}
\newcommand{\raptora}{\ensuremath{A^\texttt{RAPTOR}}}

\newcommand{\redu}[1]{ #1^*}
\newcommand{\bigvec}[1]{\overrightarrow{#1}}

\newcommand{\tpdp}{0}
\newcommand{\arxiv}{1}

\def\version{\arxiv}

\title{Manipulation Attacks in Local Differential Privacy}
\author[1]{Albert Cheu}
\author[2]{Adam Smith}
\author[1]{Jonathan Ullman}
\affil[1]{Khoury College of Computer Sciences, Northeastern University}
\affil[2]{Department of Computer Science, Boston University}

\date{}

\begin{document}

 \maketitle

\begin{abstract}
    Local differential privacy is a widely studied restriction on distributed algorithms that collect aggregates about sensitive user data, and is now deployed in several large systems.  We initiate a systematic study of a fundamental limitation of locally differentially private protocols:~\emph{they are highly vulnerable to adversarial manipulation.}  While any algorithm can be manipulated by adversaries who lie about their inputs, we show that any non-interactive locally differentially private protocol can be manipulated to a much greater extent.  Namely, when the privacy level is high or the input domain is large, an attacker who controls a small fraction of the users in the protocol can completely obscure the distribution of the users' inputs. We also show that existing protocols differ greatly in their resistance to manipulation, even when they offer the same accuracy guarantee with honest execution.  Our results suggest caution when deploying local differential privacy and reinforce the importance of efficient cryptographic techniques for emulating mechanisms from central differential privacy in distributed settings.
\end{abstract}

\ifnum\version=\tpdp
    \input{tpdp.tex}
\else
    
    \newpage
    
    \bigskip
    
    \setcounter{tocdepth}{2}
    \tableofcontents
    
    
    
    \newpage
    
    \section{Introduction}
Many companies rely on aggregates and models computed on sensitive user data. The past few years have seen a wave of deployments of systems for collecting sensitive user data via \emph{local differential privacy}~\cite{EvfimievskiGS03}, notably Google's RAPPOR~\cite{ErlingssonPK14} and Apple's deployment in \emph{iOS}~\cite{AppleDP17}.  These protocols satisfy differential privacy~\cite{DworkMNS06}, a widely studied restriction that limits the information leaked due to any one user's presence in the data. Furthermore, the privacy guarantee is enforced \emph{locally}, by a user's device, without reliance on the correctness of other parts of the system.

Local differential privacy is attractive for deployments for several reasons. The trust assumptions are relatively weak and easily explainable to novice users.  In contrast to centralized differential privacy, the data collector never collects raw data, which reduces the legal, ethical, and technical burden of safeguarding the data.  Moreover, local protocols are typically simple and highly efficient in terms of communication and computation. 

\renewenvironment{framed}[1][\hsize]
   {\MakeFramed{\hsize#1\advance\hsize-\width \FrameRestore}}%
   {\endMakeFramed}

\begin{figure}[h]
\begin{framed}[0.80\textwidth]
    \centering
    \includegraphics[width=1.0\textwidth]{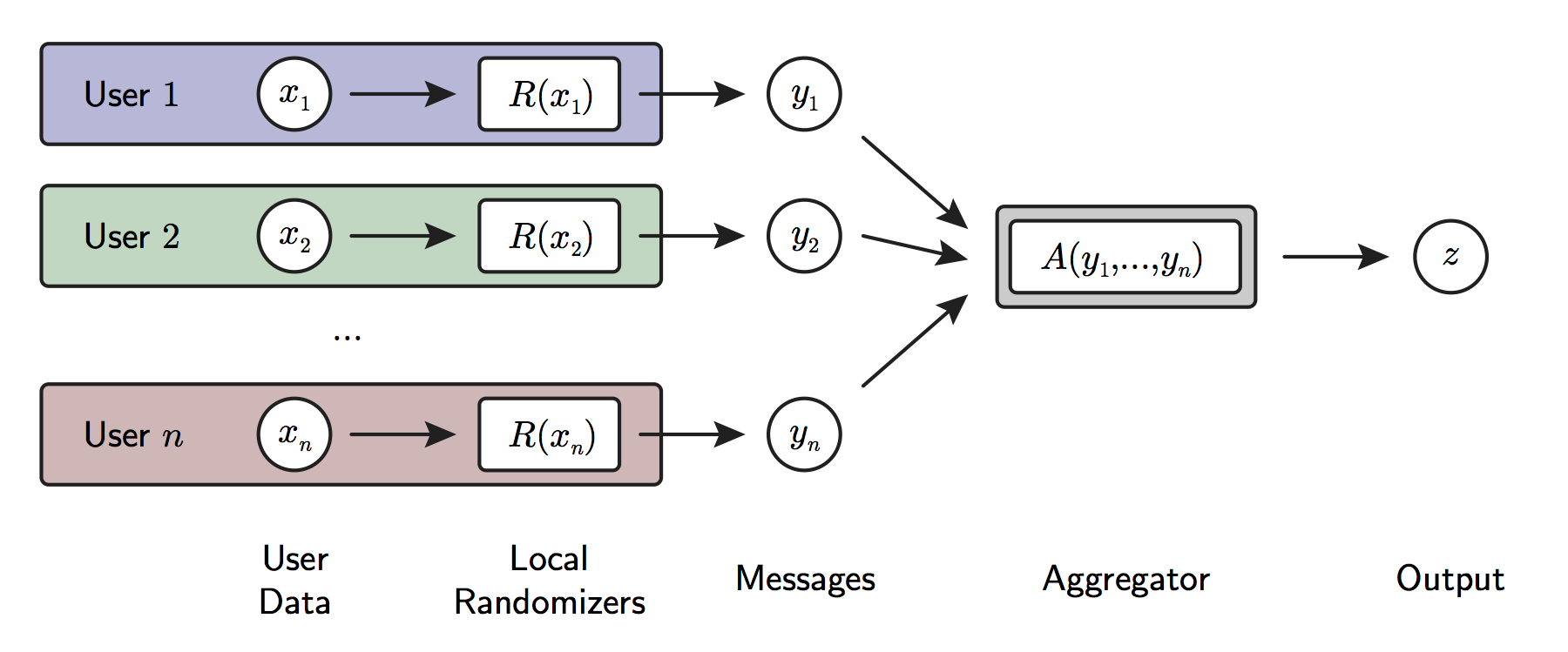}
    \vspace{-14pt}\caption{The structure of a (non-interactive) local protocol.}     \label{fig:local-model}
\end{framed}
\end{figure}

Despite these benefits, local protocols have significant limitations when compared to private algorithms in the \emph{central model}, in which data are collected and processed by a trusted curator. The most discussed limitation is larger error for the same level of privacy~(e.g., \cite{DworkMNS06,KasiviswanathanLNRS08,BeimelNO11}).  In this paper, we initiate a systematic study of a different limitation that we show to be equally fundamental:
\begin{center}
\emph{Locally differentially private protocols are highly vulnerable to manipulation.} 
\end{center}

While any algorithm can be manipulated by users who lie about their data, we demonstrate that local algorithms can be manipulated to a far greater extent.  As the level of privacy or the size of the input domain increase, an adversary who corrupts a vanishing fraction of the users can effectively prevent the protocol from collecting any useful information about the data of the honest users.  This result can be interpreted as showing that local differential privacy opens up new, more powerful avenues for \emph{poisoning attacks}---poisoning the private messages can be far more destructive than poisoning the data itself.

Prior work had already noted that a \emph{specific} protocol---Warner's randomized response~\cite{Warner65}---is vulnerable to manipulation~\cite{AmbainisJL04, MoranN06}. In contrast, our work shows that manipulation is unavoidable for \emph{any} noninteractive local protocol that solves any one of a few basic problems to sufficiently high accuracy, and systematically study the optimal degree of manipulation of local protocols for each problem.  These problems capture computing means and histograms, identifying heavy-hitters, and estimating the distribution of users' data.  We also show that existing protocols differ greatly in their vulnerability to manipulation.

Our results suggest caution when deploying locally differentially private protocols. If manipulation is a potential concern, then there should be mechanisms for enforcing the correctness of users' randomization (for instance, via software attestation). Our results also reinforce the importance of efficient cryptographic techniques that emulate central-model algorithms in a distributed setting, such as multiparty computation~\cite{DworkKMMN06} or shuffling~\cite{Bittau+17,CheuSUZZ19}.  Such protocols already have significant accuracy benefits, and our results highlight their much higher resilience to manipulation. \adamnote{maybe repeat at the end of the paper?}

\paragraph{Why are Local Protocols Vulnerable to Manipulation?}

Intuitively, because local differential privacy requires that each user's message is almost independent of their data, large changes in the users' data induce only small changes in the distribution of the messages.  As a result, the aggregator must be highly sensitive to small changes in the distribution of messages.  That is, an adversary who can cause small changes in the distribution of messages can make the messages appear as if they came from users with very different data, forcing the aggregator to change its output dramatically.

We can see how this occurs using the classic example of \emph{randomized response}.  Here each user's has data $x_i \in \{\pm 1\}$.  For roughly $2\eps$-local differential privacy, each user outputs 
\begin{equation*}
y_i = 
\begin{cases}
x_i & \textrm{with probability $\frac{1+\eps}{2}$}\\
-x_i & \textrm{with probability $\frac{1-\eps}{2}$}
\end{cases}
\end{equation*}
so that the expectation of $y_i$ is $\eps x_i $.  The aggregator can compute an unbiased estimate of the mean $\frac{1}{n} \sum_{i=1}^{n} x_i$ by returning
$
\frac{1}{n} \sum_{i=1}^{n} \frac{y_i}{\eps} \,.
$

In order to extract the relatively weak signal and make the estimate unbiased, the aggregator scales up each message $y_i$ by a factor of $\frac{1}{\eps}$, which increases the influence of each message.  Specifically, an adversary who can flip $m$ of the messages $y_i$ from $-1$ to $+1$ will increase the aggregator's output by $\frac{2m}{\eps n}$.  A simple consequence of our work is that \emph{any} noninteractive LDP protocol for computing the average of bits is similarly vulnerable to manipulation.

\subsection{A Representative Example: Frequency Estimation}

We can more fully illustrate our work results through the example of \emph{frequency estimation}.  
Consider a protocol whose goal is to collect the frequency of words typed by users on their keyboard.  We assume that there are $n$ users, and each user contributes only a single word to the dataset, so each user's word is an element of $[d] = \{1,\dots,d\}$ where $d$ is the size of the dictionary.  The goal of the protocol to estimate the vector consisting of the frequency of each word as accurately as possible.  In this example, we measure accuracy in the $\ell_1$ norm (or, equivalently, in statistical distance or total variation distance): if $v \in \R^d$ is the frequency vector whose entries $v_j$ are the fraction of users whose data takes the value $j$, and $\hat{v}$ is the estimated frequency vector, then the error is $\| v - \hat{v} \|_1 = \sum_{j=1}^{d} | v_j - \hat{v}_j |$.

\begin{wrapfigure}[13]{R}{0.48\textwidth}
  \vspace{-20pt}
\begin{framed}
    \centering
    \includegraphics[width=1.0\textwidth]{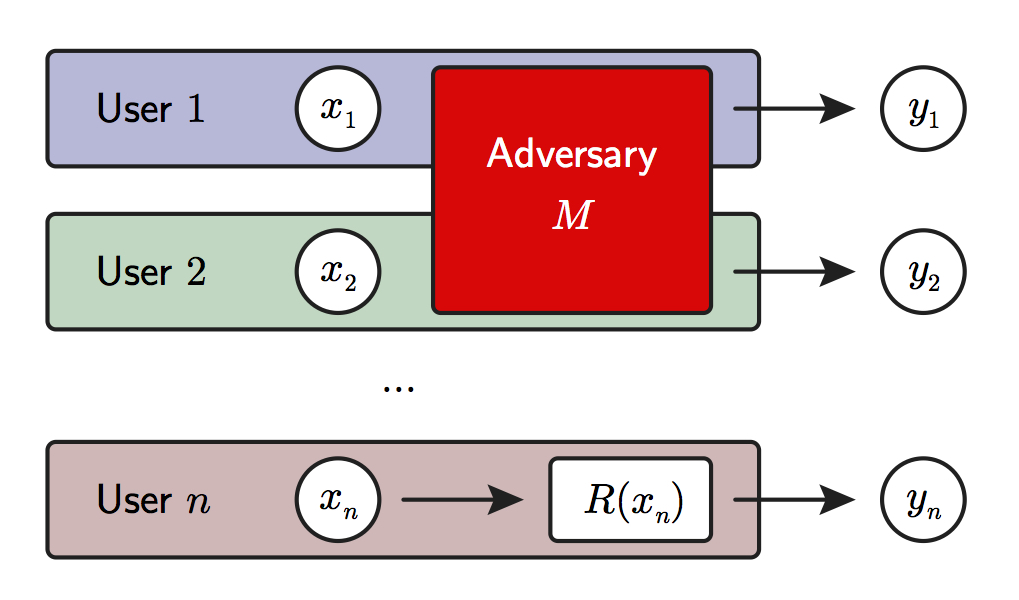}
    \vspace{-14pt}\caption{A general manipulation attack.}     \label{fig:generic-attack}
\end{framed}
\end{wrapfigure}

We consider a general attack model where the adversary is able to corrupt a set of $m$ out of the $n$ users' devices, and can instruct these users to send arbitrary messages, possibly in a coordinated fashion.  The corruptions are unknown to the aggregator running the protocol to prevent the aggregator from ignoring the messages of the corrupted users.
 In this, and all of our examples, the adversary's goal is to make the error as large as possible---exactly opposite to goal of the protocol. 

\medskip\noindent\textbf{Baseline Attacks.}
In order for the attack to be a concern, the adversary has to be able to introduce more error than what would otherwise exist in the protocol, and the attack should be specific to local differential privacy.  In particular, we say the attack is non-trivial if it introduces more error than the following trivial \emph{baselines}:

\smallskip \emph{No Manipulation.}
The adversary could choose not to manipulate the messages at all, in which case the protocol will still incur some error due to the fact that it must ensure local differential privacy.  For example, it is known that an optimal $\eps$-differentially private local protocol for frequency estimation introduces error $\approx \sqrt{d^2 / \eps^2 n}$~\cite{DuchiJW16}.

\smallskip\emph{Input Manipulation.}\anote{Used the term ``poisoning'' here}
The adversary could have the corrupted users change only their inputs.  That is, the corrupted users could honestly carry out the protocol as if their data were some arbitrary $x'_i$ instead of $x_i$ (see Figure~\ref{fig:input-manip}). Since the corrupted users control an $m/n$ fraction of the data, they can skew the overall distribution by $m/n$.  This attack applies to \emph{any} protocol, private or not.
\begin{wrapfigure}[13]{R}{0.48\textwidth}
\begin{framed}
    \centering
    \includegraphics[width=1.0\textwidth]{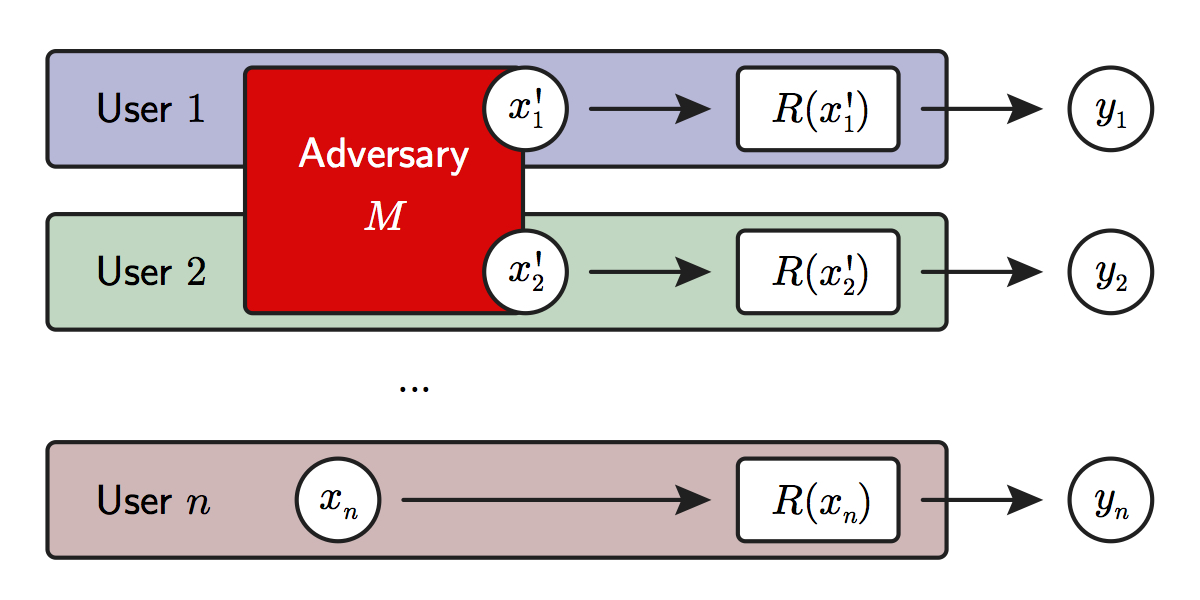}
    \vspace{-14pt}
    \caption{An input-manipulation attack.}     \label{fig:input-manip}
\end{framed}
\end{wrapfigure}

\smallskip These baselines make sense in the context of any task, and we will use the bounds for these baselines to calibrate the effectiveness of attacks for other problems (not just frequency estimation) in the next section.

\medskip\noindent\textbf{Our Manipulation Attack.}
In Section~\ref{sec:finite-attack}, we describe and analyze an attack that skews the overall distribution by $\approx \frac{m \sqrt{d}}{\eps n}$, for \emph{any} noninteractive $\eps$-differentially private local protocol.  This attack introduces much larger error---by roughly a factor of $\frac{\sqrt{d}}{ \eps}$---than input manipulation, and thus shows specifically that locally private protocols are highly vulnerable to manipulation.  Moreover, when the adversary corrupts $m \approx \sqrt{nd}$ (which is much smaller than $n$ for any interesting choice of $n$ and $d$) users, then they can significant reduce the accuracy of the protocol.  
We can also show that our attack is near-optimal by demonstrating that achieves optimal error in the absence of manipulation, and cannot be manipulated by more than $\approx \frac{m\sqrt{d}}{\eps n}$.

\medskip\noindent\textbf{The Breakdown Point.}
Another way to understand the effectiveness of a manipulation attack is through what we call the \emph{breakdown point}---the maximum fraction of corrupted users that any protocol can tolerate while still ensuring non-trivial accuracy.  Our attack demonstrates that, for frequency estimation, the breakdown point is roughly $\frac{\eps}{\sqrt{d}}$.  That is, that this number of corrupted users can skew the distribution by $\Omega(1)$ in $\ell_1$ norm, while \emph{any} two frequency vectors have $\ell_1$ distance at most 1.  Thus, when $\eps$ is small or $d$ is large, an attacker who controls just a vanishing fraction of the users can prevent the protocol from achieving any non-trivial accuracy guarantee.

\subsection{Summary of Results}

In this work, we construct two manipulation attacks on locally differentially private protocols, and use these attacks to derive lower bounds on the degree of manipulation allowed by local protocols for a variety of tasks (including the frequency estimation example above). We also study the resilience of specific protocols to manipulation. For each problem, we give a protocol that is asymtpotically optimal with respect to both ordinary accuracy (i.e., without manipulation) and resilience to manipulation. We also show that popular protocols for most tasks are much less resistant to manipulation than optimal ones. 

Below, we first discuss the attacks informally, and then discuss the set of problems to which they apply.  We defer details of the attack model to  Section~\ref{sec:threatmodel}.  Our results are summarized in Table~\ref{tab:results}.

\medskip\noindent\textbf{Manipulation Attacks for Binary Data.}
Our first attack concerns the simplest problem in local differential privacy---computing a mean of bits.  Each user has data $x_i \in \zo$, and we assume that each $x_i$ is drawn independently from the Bernoulli distribution $\Ber(p)$, meaning $x_i = 1$ with probability $p$ and $x_i = 0$ with probability $1-p$.  Our goal is to estimate the mean $p$ as accurately as possible.  More generally, we could allow the users to have arbitrary data $x_1,\dots,x_n \in \zo$ and try to estimate $\frac{1}{n} \sum_{i=1}^{n} x_i$.  For the purposes of attacks, considering the former distributional version will make our results stronger.

Without manipulation, this problem is solved by the classical randomized response protocol~\cite{Warner65}, which achieves optimal error $\Theta(\frac{1}{\eps \sqrt{n}})$.  As we discussed in the introduction, one can show that the error of randomized response increases to $\Theta(\frac{1}{\eps \sqrt{n}} + \frac{m}{\eps n})$ when an adversary corrupts $m$ of the users. We show that no protocol can improve this bound.
\begin{thm}[Informal]
For every $\eps$-differentially private local protocol $\Pi$ for $n$ users with input domain $\zo$, there is an attack $M$ corrupting $m$ users such that $\Pi$ cannot distinguish between the following cases:
\begin{enumerate}
    \item The data is drawn from $\Ber(p_0)$ for $p_0 = \frac12$ and $\Pi$ has been manipulated by $M$.
    \item The data is drawn from $\Ber(p_1)$ for $p_1 = \frac12 + \Theta(\frac{1}{\eps}(\frac{1}{\sqrt{n}} + \frac{m}{n}))$ and $\Pi$ has not been manipulated.
\end{enumerate}
\end{thm}
This theorem shows that, when the data is drawn from $\Ber(p)$ for unknown $p$, no protocol $\Pi$ can estimate $p$ and guarantee accuracy better than $\Theta(\frac{1}{\eps \sqrt{n}} + \frac{m}{\eps n})$.  As an immediate consequence, when the data $x_1,\dots,x_n \in \zo$ may be arbitrary, no protocol $\Pi$ can estimate the mean $\frac{1}{n} \sum_i x_i$ with higher accuracy.

\medskip\noindent\textbf{Manipulation Attacks for Large Domains.}
Since estimating the mean of bits is a special case of most problem studied in the local model, this attack already shows that manipulation can cause additional error of $\Omega(\frac{m}{\eps n})$ for many problems.  In some cases, this bound is already near-optimal, and some protocol achieves a similar upper bound.  However, for many cases of interest (such as the frequency estimation example), protocols become more vulnerable to manipulation when the size of the input domain increases.  Our second result is an attack on any protocol accepting inputs from the domain $[d] = \{1,\dots,d\}$ for large $d$, showing that manipulation can skew the distribution by $\tilde{\Omega}(\frac{m \sqrt{d}}{\eps n})$ without being detected. 
\begin{thm}[Informal] \label{thm:main-intro-1}
For every $\eps$-differentially private local protocol $\Pi$ for $n$ users with input domain $[d]$, there is an attack $M$ corrupting $m$ users such that $\Pi$ cannot distinguish between the following cases:
\begin{enumerate}
    \item The data is drawn from the uniform distribution $\bU$ over $[d]$ and $M$ manipulates $\Pi$.
    \item The data is drawn from some distribution $\bP$ over $[d]$ with $\| \bU - \bP \|_1 = \Theta\left(\frac{1}{\eps}\sqrt{\frac{d}{\log n}}\left(\frac{1}{\sqrt{n}} + \frac{m}{n}\right)\right)$ and $\Pi$ has not been manipulated.
\end{enumerate}
For a large class of natural protocols, the bound on $\| \bU - \bP \|_1$ can be sharpened to $\Theta(\frac{\sqrt{d}}{\eps}(\frac{1}{\sqrt{n}} + \frac{m}{n}))$.
\end{thm}
A consequence of this attack for the example of frequency estimation above is that any local protocol can have the distribution skewed by $\tilde{\Omega}(\frac{m\sqrt{d}}{\eps n})$.  As we show in Section~\ref{sec:protocols}, this bound is actually matched by a simple protocol.
\adamnote{Maybe have a techniques paragraph here?}

\medskip\noindent\textbf{Applications.}
We consider a variety of tasks of interest in local differential privacy, and show that for each of them, one of the two attacks above gives an optimal bound on how vulnerable protocols for that task are to manipulation.  The results are summarized in Table~\ref{tab:results}.

Most tasks we consider can be formulated as instances of the following \emph{$\ell_p / \ell_q$-mean estimation} problem for vectors in $\R^d$.\footnote{Given any vector $v \in \R^d$ and any $p \geq 1$, the $\ell_p$-norm is defined as $\|x\|_p = (\sum_{j=1}^{d} |x_j|^p)^{1/p}$.  For $p = \infty$, the $\ell_\infty$ norm is defined as $\max_{j = 1,\dots,d} |x_j|$.}  Each user's data $x_i$ is a vector in $\R^d$ such that the $\ell_p$-norm of each data point is bounded, $\|x_i \|_p \leq 1$.  The protocol's goal is to output an estimate of the mean $\hat\mu$ with low error in the $\ell_q$-norm, $\| \hat\mu - \frac{1}{n} \sum_{i=1}^{n} x_i \|_q$.  This setup captures a number of widely studied problems:
\begin{table}[t!]
\begin{center}
\begin{tabular}{|l|c|c|c|c|}
\hline
\multicolumn{1}{|c|}{Problem} & No Manip. & Manip. UB & Manip. LB & Breakdown Point \\
\hline
$\ell_1/\ell_1$ Estimation & $\Theta(\sqrt{\frac{d^2}{\eps^2 n}})$ & $\tilde{O}(\frac{m}{n} \cdot\frac{\sqrt{d}}{\eps})$ & $\Omega(\frac{m}{n} \cdot\frac{\sqrt{d}}{\eps \sqrt{\log n}})$~$\flat$ & \multirow{2}{*}{$O\left(\eps \sqrt{\frac{\log n}{d}} \right)$~$\flat$} \\
(Frequency Estimation) & \cite{DuchiJW16} & Thm \ref{thm:esto-l1-informal} & Thm \ref{thm:estimation-attack} &  \\ \hline
$\ell_1/\ell_1$ Testing & $\Theta(\sqrt{\frac{d}{\eps^2 n}})$ & $O(\frac{m}{n} \cdot\frac{\sqrt{d}}{\eps })$ & $\Omega(\frac{m}{n} \cdot\frac{\sqrt{d}}{\eps \sqrt{\log n}})$~$\flat$ & \multirow{2}{*}{$O\left(\eps \sqrt{\frac{\log n}{d}} \right)$~$\flat$} \\
(Uniformity Testing) & \cite{AcharyaCFT19} & Thm \ref{thm:raptor} & Thm \ref{thm:uniformity-attack} &  \\ \hline
$\ell_1/\ell_\infty$ Estimation & $\Theta(\sqrt{\frac{\log d}{\eps^2 n}})$ & $O(\frac{m}{n} \cdot\frac{\log d}{\eps})$ & $\Omega(\frac{m}{n} \cdot\frac{1}{\eps })$ & \multirow{2}{*}{$O(\eps)$} \\
(Histograms / HH) & \cite{BassilyS15} & Thm \ref{thm:hh-under-attack} & Thm \ref{thm:binary-attack} &  \\ \hline
$\ell_\infty/\ell_\infty$ Estimation & $\Theta(\sqrt{\frac{d \log d}{\eps^2 n}})$ & $O(\frac{m}{n} \cdot \frac{1}{\eps})$~$\sharp$ & $\Omega(\frac{m}{n}\cdot\frac{1}{\eps})$ & \multirow{2}{*}{$O(\eps)$} \\
($d$ Statistical Queries) & [Folklore] & Thm \ref{thm:esti-under-attack} & Thm \ref{thm:binary-attack} &  \\ \hline
$\ell_2/\ell_2$ Estimation & $\Theta(\sqrt{\frac{d}{\eps^2 n}})$ &  $\tilde{O}(\frac{m}{n}\cdot\frac{1}{\epsilon})$ & $\Omega(\frac{m}{n}\cdot\frac{1}{\eps})$ & \multirow{2}{*}{$O(\eps)$} \\
(Gradients) & \cite{DuchiJW13} & Thm~\ref{thm:estt-under-attack} & Thm \ref{thm:binary-attack} &  \\ \hline
\end{tabular}
\caption{Summary of Results.  In each case, [No Manipulation] is the optimal error achievable under local differential privacy without manipulation.  For each problem, we identify some protocol that has optimal error without manipulation such that manipulation can increase the error by [Manipulation UB] and show that manipulation can make the error of any local protocol as large as [Manipulation LB].  Finally, in each case, no protocol can guarantee non-trivial accuracy in the presence of [Breakdown Point] corrupted users.  $\sharp$ indicates that the upper bound limited to public-string-oblivious attacks.  $\flat$ indicates that the $\sqrt{\log n}$ factor can be removed for a natural class of protocols.  In all cases, lying about $m$ inputs (``input manipulation'') influences the correct output by $\frac m n$; we present the upper and lower bounds as multiples of that baseline.}
\label{tab:results}
\end{center}
\end{table}
\begin{itemize}
    \item The frequency estimation example above is a special case of $\ell_1/\ell_1$ estimation, where each user represents their word $x_i \in [d]$ by the standard basis vector $e_{x_i} \in \R^d$ with a $1$ in the $x_i$-th coordinate and $0$ elsewhere.
    
    \item Computing a histogram of data in $[d]$ is a special case of $\ell_1 / \ell_\infty$-mean estimation. The \emph{heavy-hitters (HH)} problem, which asks one only to identify the heaviest bins of a histogram and their frequencies, suffices to solve $\ell_1/\ell_\infty$-mean estimation, so manipulation attacks on the latter thus imply attacks on the former.  
      Computing heavy-hitters has been a focal point in the past few years~\cite{HsuKR12,BassilyS15,BassilyNST17,BunNS18}, and it is central to systems deployed by Google and Apple~\cite{ErlingssonPK14,AppleDP17}.
    
    \item Computing the answers to $d$ statistical queries~\cite{Kearns93,BlumDMN05,KasiviswanathanLNRS08} is 
      a special case of $\ell_\infty/\ell_\infty$-mean estimation.  Users have data in some arbitrary domain $\cX$, there are $d$ query functions $f_1,\dots,f_d \from \cX \to [-1,1]$, and we would like an accurate estimate of $\sum_{i=1}^{n} f_j(x_i)$ for every $j$.  In the corresponding mean estimation instance, $x_i=(f_1(x_i),\dots,f_d(x_i))$. 
    
    \item When minimizing a sum of convex functions $f(\theta) = \sum_{i=1}^{n} f_{x_i}(\theta)$ defined by the users' data (e.g.~to train a machine learning model), one often computes the average gradient $\sum_{i=1}^{n} \nabla f_{x_i}(\theta_t)$ at a sequence of points $\theta_t$.  Typically one assumes that the gradients are bounded in $\ell_2$, and convergence requires an accurate estimate in $\ell_2$, making this an instance of $\ell_2/\ell_2$-mean estimation.  (More generally, optimization \emph{requires} this sort of estimation~\cite{BassilyST14}). 
    
    \item We consider one further problem, \emph{$\ell_1 / \ell_1$-uniformity testing}, for which Acharya et al.~\cite{AcharyaCFT19} considered LDP protocols. Assuming the data is drawn from some distribution over $[d]$, we want to determine if this distribution is either uniform or is far from uniform in $\ell_1$ distance.  
\end{itemize}

Since every $\ell_p/\ell_q$ mean estimation problem generalizes binary mean estimation (the special case where $d = 1$), our first attack gives a lower bound on all of these problems.  Our second attack is precisely an attack on the $\ell_1/\ell_1$-testing problem, and thus implies a lower bound of $\tilde{\Omega}(\frac{m\sqrt{d}}{\eps n})$ for that problem.  Finally, since $\ell_1/\ell_1$-mean estimation problem strictly generalizes $\ell_1/\ell_1$-testing problem---once we estimate the mean, we can determine if it is close to uniform or far from uniform---we obtain the same lower bound for that problem.

\medskip\noindent\textbf{Resilient Protocols} \adamnote{added this.} For all of these problems we also identify and analyze protocols whose error nearly matches the lower bounds established by our attacks. These protocols generally use the public-coin model to compress each player's report to a single bit, thus reducing their influence. 

While all of our optimal protocols were known prior to our work, we demonstrate that the choice of protocol is crucial. Some well known protocols with optimal accuracy absent manipulation allow for much greater manipulation than necessary.  For example, the simplest adaptation of randomized response to frequency estimation, in which each player sends one bit per potential item, allows $m$ corrupted users to introduce error about $m d/\epsilon n$ in a direction of their choice, which is about $\sqrt{d}$ larger than optimal.

\subsection{Related Work}

Prior work had already observed that the specific randomized response protocol was vulnerable to manipulation~\cite{AmbainisJL04, MoranN06}.  In contrast to ours, these works constructed efficient cryptographic protocols for sampling from the correct distribution, which resist our attacks.  Our work shows that some degree of cryptography is necessary to avoid manipulation.

Our work is loosely related to \emph{data poisoning attacks} in adversarial machine learning.  In data poisoning, the adversary is inserts additional data to somehow degrade the quality of the output.  Our attacks can be viewed as data poisoning attacks where the ``data'' being poisoned is actually the messages to the protocol.  Thus, our results can be viewed as showing that adding local randomization to achieve privacy makes the protocol much more vulnerable to data poisoning.

Our work is also related to the literature on \emph{robust statistics}.  In the standard model of robust statistics, we are given data drawn from distribution $\mathbf{P}$ with some structure (e.g.~$\mathbf{P}$ is a Gaussian distribution), but some small fraction of the data has been corrupted with arbitrary data, and the goal to identify the distribution $\mathbf{P}$ as well as possible.  Our setting is similar except that we don't get access to the data directly, but only once its been filtered through some set of private local randomizers.  One might hope to obtain local protocols that are robust to manipulation using techniques from robust statistical estimators on the distribution of messages induced by the local randomizers.  Our attacks can be viewed as showing that such robust estimators don't exist.

\subsection{Organization}
In Section~\ref{sec:prelims} we introduce the model and key concepts.  In Section~\ref{sec:binary-attack}, we demonstrate attacks on protocols for binary data, and in Section~\ref{sec:finite-attack}, we demonstrate attacks on protocols for large data domains.  In Section~\ref{sec:protocols} we identify protocols with near-optimal resistance to manipulation for a variety of canonical problems in local differential privacy.  In Section~\ref{sec:badprotocols} we highlight the fact that not all protocols with optimal error absent manipulation are optimally robust to manipulation.


    
    \section{Threat Model and Preliminaries}
\label{sec:prelims}

\subsection{Local Differential Privacy}
In this model there are $n$ \emph{users}, and each user $i \in [n]$ holds some sensitive data $x_i \in \cX$ belonging to some \emph{data universe} $\cX$. There is also a public random string $S$. Finally there is a single \emph{aggregator} who would like to compute some function of the users' data $x_1,\dots,x_n$.  In this work, for simplicity, we restrict attention to \emph{non-interactive local differential privacy}, meaning the users and the aggregator engage in the following type of protocol:
\begin{enumerate}
\item A public random string $S$ is chosen from some distribution $\bS$ over support $\cS$.
\item Each user computes a \emph{message} $y_i \gets R_i(x_i, b)$ using a \emph{local randomizer} $R_i \from \cX \times \cS \to \cY$.
\item The aggregator $A \from \cY^n \times \cS \to \cZ$ computes some output $z \gets A(y_1,\dots,y_n, S)$.
\end{enumerate}
Thus the protocol $\Pi$ consists of the tuple $\Pi = ((R_1,\dots,R_n), A, \bS)$.  We will sometimes write $\vec{R}$ to denote the local randomizers $(R_1,\dots,R_n)$.  If $R_1 = \dots = R_n = R$ then we say the protocol is \emph{symmetric} and denote it $\Pi = (R, A, \bS)$.

Given user data $\vec{x} \in \cX^n$ we will write $\Pi(\vec{x})$ to denote the distribution of the protocol's output when the users' data is $\vec{x}$, and $\vec{R}(\vec{x})$ denotes the distribution of the protocol's messages.  Given a distribution $\bP$ over $\cX$, we will write $\Pi(\bP)$
and $\vec{R}(\bP)$ to denote the resulting distributions when $\vec{x}$ consists of $n$ independent samples from $\bP$.

Informally, we say that the protocol satisfies \emph{local differential privacy}~\cite{EvfimievskiGS03,DworkMNS06,KasiviswanathanLNRS08} if the local randomizers depend only very weakly on their inputs.  Formally,
\begin{defn}[Local DP~\cite{EvfimievskiGS03,DworkMNS06,KasiviswanathanLNRS08}]
A protocol $\Pi=((R_1,\dots, R_n), A, \bS)$ satisfies \emph{$(\eps,\delta)$-local differential privacy} if for every $i \in [n]$, every $x,x' \in \cX$, every $S \in \cS$ and every $Y \subseteq \cY$,
\begin{equation}
\label{eq:dp}
\pr{R_i}{R_i(x, S) \in Y} \leq e^\eps \cdot \pr{R_i}{R_i(x', S) \in Y} + \delta
\end{equation}
where we stress that the randomness is \emph{only} over the coins of $R_i$ and \emph{not} over the randomness of $S$.  If $\delta=0$, then we will simply write \emph{$\eps$-local differential privacy}.
\end{defn}

\subsection{Threat Model: Manipulation Attacks} \label{sec:threatmodel}
We capture manipulation attacks via a game involving a protocol $\Pi = (\vec{R}, A, \bS)$, a vector $\vec{x}$ of $n$ data values, and an adversary $M$.  We parameterize the game by the number of users $n$ and the number of corrupted users $m \leq n$, written as $\mathrm{Manip}_{m,n}$; when clear from context, the subscript is omitted.  The crux of the game is that the adversary corrupts a set $C$ of at most $m$ users, then the users are assigned data $\vec{x}$, and then either play \emph{honestly} by sending the message $y_i = R_i(x_i,b)$ or they \emph{manipulate} by playing some arbitrary message chosen by the adversary. Figure \ref{fig:generic-attack} presents the structure of an attack in the case where $C=\{1,2\}$.

The game is described in Figure~\ref{fig:manipulation}, including a possible restriction on the attacker. We use $\mathrm{Manip}_{m,n}(\Pi, \vec{x}, M)$ to denote the distribution on outputs of the protocol on data $\vec{x}$ and messages manipulated by $M$, and $\mathrm{Manip}_{m,n}(\vec{R}, \vec{x}, M)$ to denote the distribution of messages in the protocol. Given a distribution $\bP$ over $\cX$, we will use $\mathrm{Manip}_{m,n}(\Pi, \bP, M)$ and $\mathrm{Manip}_{m,n}(\vec{R}, \bP, M)$ to denote the resulting distributions when $\vec{x}$ consists of $n$ independent samples from $\bP$.

\begin{figure}[h!] 
\begin{framed}
\textbf{Parameters:} $0 \leq m \leq n$. \\
\textbf{Elements:} A protocol $\Pi = (\vec{R}, A, \bS)$ for $n$ users, a vector of data $\vec{x}$, an attacker $M$.
\begin{enumerate}

	\item Each user $i$ is given data $x_i$.

    \item The public string $S \sim \bS$ is sampled.
    
    \item The attacker $M$ chooses a set of \emph{corrupted users} $C \subseteq [n]$ of size $m$.
        \begin{itemize}
    		\item[] If the corruptions are independent of the public string $S$ then they are \emph{public-string-oblivious}, and otherwise they are \emph{public-string-adaptive}. 
    	\end{itemize}
	
   \item The attacker $M$ chooses a set of messages $\{ y_i \}_{i \in S}$ for the corrupted users.
   
   \item The non-corrupted users $i \not\in C$ choose messages $y_i \sim R_i(x_i, b)$ honestly.
   
   \item The aggregator returns $z \gets A(y_1,\dots,y_n, b)$.
   \end{enumerate}
\end{framed}
\caption{Manipulation Game $\mathrm{Manip}_{m,n}$} \label{fig:manipulation}
\end{figure}

\subsection{Notational Conventions}
Throughout, boldface roman letters indicate distributions (e.g.~$\mathbf{P}$).  Vectors are denoted $\vec{v} = (v_1,v_2,\dots)$. We write $[n]$ to denote the set $\{1,\dots, n\}$. $\Rad(\mu)$ will denote the distribution over $\{\pm 1\} $ with mean $\mu$ (e.g. $\Rad(0)$ is the standard Rademacher distribution).

    
    \section{Attacks Against Protocols for Binary Data}
\label{sec:binary-attack}

In this section, we show how to attack any protocol that estimates the mean of a Rademacher distribution $\Rad(\mu)$. \footnote{The choice of data universe $\cX = \{\pm 1\}$ simplifies the analysis but is not inherent to the results; any binary data universe has counterpart attacks.} In particular, we show that any such protocol has error $\Omega(m/\eps n + 1/\eps\sqrt{n})$ in the presence of $m$ corrupt users. The proof has two main steps. In the first, we argue that every $\eps$ differentially private protocol $\Pi$ for binary data is equivalent \emph{with respect to manipulation} to a protocol $\redu{\Pi}$ where each user applies randomized response, building on \cite{KairouzOV15}. That is, an attack against $\redu{\Pi}$ implies an attack against $\Pi$. In the second step, we construct an attack against any randomized response protocol and show that it makes two distributions with distance $\approx m/\eps n + 1/\eps\sqrt{n}$ indistinguishable to the protocol.

\subsection{Reduction to Randomized Response}
This subsection will show that it is without loss of generality to design attacks for the family of \emph{$\eps$-randomized response} protocols, in which each user's randomizer is $\rrr_\eps$ (see \eqref{eq:rrr} below) but the aggregator is arbitrary.

\begin{equation}
    \label{eq:rrr}
    \rrr_\eps(x) := \begin{cases} x&\mathrm{with~probability}~\frac{e^\eps}{e^\eps + 1}\\
    -x&\mathrm{with~probability}~\frac{1}{e^\eps + 1}
    \end{cases}
\end{equation}

\cite{KairouzOV15} established that $\rrr_\eps$ dominates any other $\eps$-private local randomizer for binary data; below, we present the result in a syntax more similar to that of \cite{MurtaghV18}.
\begin{lem}[\cite{KairouzOV15}]
\label{lem:rr}
For any $\eps$-private randomizer $R: \{\\pm 1\} \rightarrow \cY$, there exists a randomized algorithm $\redu{R}$ such that, for any $x \in \{\pm 1\}$, $\redu{R}(\rrr_\eps(x))$ and $R(x)$ are identically distributed.
\end{lem}

\begin{algorithm}
\caption{$\redu{A}_{\vec{R}, A}: \{\pm 1\}^n \to \cZ$}
\label{alg:transform-agg}
\Parameters{Vector of randomizers $\vec{R} = (R_1, \dots, R_n)$; aggregator $A:\cY^n \rightarrow \cZ$}

\KwIn{$\vec{y} \in \{\pm 1\}^n$}
\KwOut{$z \in \cZ$}
\medskip

For each $i\in [n]$, construct $\redu{R_i} : \{\pm 1\} \rightarrow \cY$ from $R_i$ as guaranteed by Lemma \ref{lem:rr}

Sample $z \sim A ( \redu{R}_1(y_1),\dots, \redu{R}_n(y_n) )$

\Return{$z$}
\end{algorithm}

For a vector of randomizers $\vec{R}=(R_i)_{i\in [n]}$, we define the vector $\vec{\redu{R}} := (\redu{R_i})_{i\in [n]}$. Fix any $\eps$-locally private protocol $\Pi = ((R_1, \dots, R_n), A)$ with data universe $\{\pm 1\}$ and message universe $\cY$. We will use $\redu{\Pi}$ to denote the symmetric protocol that uses randomizer $\rrr_\eps$ and aggregator $\redu{A}_{\vec{R},A}$ (see Algorithm \ref{alg:transform-agg}). Lemma \ref{lem:rr} directly implies that the transformation preserves any guarantees about the output (e.g. estimates of a Rademacher parameter will have the same error) in the absence of an attack:
\begin{coro}
\label{coro:reduction-to-rr-1}
For any $\eps$-locally private protocol $\Pi=((R_1,\dots, R_n), A)$ for binary data and any mean $\mu \in [-1,+1]$, $\Pi(\Rad(\mu))$ and $\redu{\Pi}(\Rad(\mu))$ are identically distributed.
\end{coro}

Now we claim that we can adapt an attack against $\redu{\Pi}$ into one against $\Pi$, again preserving any guarantees about the output.
\begin{thm}
\label{thm:reduction-to-rr-2}
Fix any $\eps$-locally private protocol $\Pi=((R_1,\dots, R_n), A)$ for binary data and any $m\leq n$. For any manipulation attack $\redu{M}$ against $\redu{\Pi}$, there exists an attack $M$ against $\Pi$ with the following property: for any mean $\mu \in [-1,+1]$, $\mathrm{Manip}_{m,n}(\Pi, \Rad(\mu), M)$ and $\mathrm{Manip}_{m,n}(\redu{\Pi}, \Rad(\mu), \redu{M})$ are identically distributed.
\end{thm}
\begin{proof}
Define $M$ to be the attack that first executes $\redu{M}$ to obtain a set of corrupt users $C$ and their messages $(\redu{y}_i)_{i \in C}$ which lies in $\{\pm 1\}^m$. Then it generates $(y_i)_{i \in C}$, which lies in $\cY^m$, by executing $\redu{R}_i(\redu{y}_i)$ for every $i \in C$.

To ease the presentation of the proof, we assume without loss of generality that users are sorted so that the corrupted users $C$ consist of the first $m$ users. We use $\redu{\bY}_{[m]}$ to denote the distribution of $(\redu{y}_i)_{i \in [m]}$ and $\redu{R}_{[m]}(\redu{\bY}_{[m]})$ to denote the distribution of $(y_i)_{i \in [m]}$. Below, each step equates distributions:
\begin{align*}
\mathrm{Manip}(\Pi, \Rad(\mu), M) &= A \left( \redu{R}_{[m]}(\redu{\bY}_{[m]}) \times R_{m+1}(\Rad(\mu)) \times \dots \times R_n(\Rad(\mu)) \right)\\
    &= A\left( \redu{R}_{[m]}(\redu{\bY}_{[m]})\times \redu{R}_{m+1}(\rrr_\eps(\Rad(\mu))) \times \dots \times \redu{R}_n(\rrr_\eps(\Rad(\mu))) \right)\\
    &= (A \circ \vec{\redu{R}})\bigg(\redu{\bY}_{[m]} \times \underbrace{\rrr_\eps(\Rad(\mu)) \times \dots \times \rrr_\eps(\Rad(\mu))}_{n-m~\mathrm{copies}} \bigg)\\
    &= \mathrm{Manip}(\redu{\Pi}, \Rad(\mu), \redu{M})
\end{align*}

The second equality comes from Lemma \ref{lem:rr}. This concludes the proof.
\end{proof}

\subsection{The Attack on Randomized Response}

In this section, we describe an attack against $\eps$-randomzied response protocols. By Theorem \ref{thm:reduction-to-rr-2}, the statements we prove will generalize to arbitrary protocols. In particular, no protocol will be able to distinguish between (1) the scenario where there is no attack and data comes from $\Rad(\mu)$ for a particular choice of mean $p \approx \frac{m}{\eps n} + \frac{1}{\eps \sqrt{n}}$, and (2) the scenario where our attack is present and data comes from $\Rad(0)$. This will imply no protocol can estimate up to error $p/2$ in the presence of $m$ corrupt users.

\begin{figure}
\begin{framed}
Choose $C$, the users to corrupt, by uniformly sampling from all subsets of $[n]$ with size $m$

Command each corrupted user $i\in S$ to report $y_i \gets +1$.
\end{framed}

\caption{A manipulation attack $M^\rr_{m,n}$ against randomized response}
\label{fig:binary-attack}
\end{figure}

The attack $M^\rr_{m,n}$ is sketched in Figure \ref{fig:binary-attack}. To aid in the analysis, we define the functions $\mu(m,n) := \frac{m}{n} + \sqrt{\frac{2\ln 6}{n}}$ and $\mu(m,n,\eps) := \frac{e^\eps + 1}{e^\eps - 1}\cdot \mu(m,n) = \frac{e^\eps + 1}{e^\eps - 1} \big( \frac{m}{n} + \sqrt{\frac{2\ln 6}{n}} \big)$.

We will first show that $M^\rr_{m,n}$ will make the \emph{messages} generated by $\vecrrr(\Rad(0))$ indistinguishable\footnote{We remark our notion of indistinguishability is not an explicit bound on statistical distance, but instead a one-sided version of the differential privacy guarantee.} from the messages generated by $\vecrrr(\Rad(\mu(m,n,\eps)))$ (Lemma \ref{lem:binary-attack}). Because the same aggregator will be run on both sets of messages, we show that the \emph{protocol's output} is likewise indistinguishable (Corollary \ref{coro:binary-attack}).

\begin{lem}
\label{lem:binary-attack}
For any $n \geq 931$ and any $m \leq n/8$, the distribution $\vecrrr_\eps(\Rad(\mu(m,n,\eps)))$ cannot be distinguished from $\mathrm{Manip}\left(\vecrrr_\eps, \Rad(0), M^{\rr}_{m,n} \right)$ with arbitrarily low probability of failure. Specifically, for all $Y \subseteq \{\pm 1\}^n$,
\begin{equation}
\label{eq:similar-bitstrings}
\pr{}{\vecrrr_\eps(\Rad(\mu(m,n,\eps))) \in Y} \leq 51 \cdot \pr{}{\mathrm{Manip}\left(\vecrrr_\eps, \Rad( 0 ), M^{\rr}_{m,n} \right) \in Y} + \frac{1}{3}
\end{equation}
\end{lem}

\begin{proof}
We begin with the following observation: if we run $\rrr_\eps$ on a sample from $\Rad(q)$, then the output is drawn from $\Rad\left(\frac{e^\eps - 1}{e^\eps +1}\cdot q \right)$.

If we choose $q=\mu(m,n,\eps)$, the output of $\rrr_\eps$ on $\Rad(q)$ is drawn from $\Rad(\mu(m,n))$. Now consider the distribution $\vecrrr_\eps(\Rad(\mu(m,n,\eps)))$. Because it is a symmetric distribution over $\{\pm 1\}^n$---any permutation that the output takes is equally likely as any other permutation---it suffices to consider the number of bits with value $+1$. Let $W^+$ denote that number and observe that $W^+ \sim \Bin(n, \half + \half \mu(m,n) )$, where $\Bin(n,p)$ denotes the binomial distribution over $[0,n]$ with expected value $np$.

If we choose $q = 0$, a message produced by $\rrr_\eps(\Rad(q))$ is drawn from $\Rad(0)$. Now consider the distribution $\mathrm{Manip}\left(\vecrrr_\eps, \Rad( 0 ), M^{\rr}_{m,n} \right)$. Due to the random choice of corrupted users, this is a symmetric distribution over $\{\pm 1\}^n$ so it suffices to consider $W$, the number of bits with value $+1$. There are $n-m$ bits drawn from $\Rad(0)$ in addition to $m$ bits that deterministically have value $+1$, so $W \sim \Bin(n-m, \half) + m$.

In Appendix \ref{apdx:binary-attack}, we prove the technical claim below:
\begin{clm}
\label{clm:similar-binomials}
For all $n \geq 931$ and $m \leq n/8$, if we sample $W \sim m + \Bin(n-m, \half)$ and $W^+ \sim \Bin\left(n,\half + \frac{m}{2n} + \sqrt{\frac{\ln 6}{2n}} \right)$, then for any $\cW \subseteq [0,n]$,
\[
\pr{}{W^+ \in \cW} \leq 51 \cdot \pr{}{W \in \cW} + \frac{1}{3}
\]
\end{clm}

This concludes the proof.
\end{proof}

\begin{coro}
\label{coro:binary-attack}
For any $n \geq 931$, any protocol $\Pi=(\rr_\eps, n, A)$, and any $m \leq n/8$, the distribution $\Pi(\Rad(\mu(m,n,\eps)))$ cannot be distinguished from $\mathrm{Manip}\left(\Pi, \Rad(0), M^{\rr}_{m,n} \right)$ with arbitrarily low probability of failure. Specifically, if $\cZ$ is the range of $A$, then for all $Z \subseteq \cZ$,
\[
\pr{}{\Pi(\Rad(\mu(m,n,\eps))) \in Z} \leq 51 \cdot \pr{}{\mathrm{Manip}\left(\Pi, \Rad( 0 ), M^{\rr}_{m,n} \right) \in Z} + \frac{1}{3}
\]
\end{coro}
\begin{proof}
We first consider the case where $A$ is deterministic. For each $Z$ there must be some $Y\subseteq \{\pm 1\}^n$ such that $A(y)\in Z$ if and only if $y \in Y$. Hence, \eqref{eq:similar-bitstrings} implies our claim.

In the case where $A$ is randomized, we invoke the property that $A$ can be viewed as first sampling a deterministic $\hat{A}$, then returning $\hat{A}(\vec{y})$. The distribution from which $\hat{A}$ is drawn is independent of $\vec{y}$.
\end{proof}

Now we show that Corollary \ref{coro:binary-attack} implies a lower bound for Rademacher estimation.

\begin{thm}
\label{thm:binary-attack}
For any $n \geq 931$, any $\eps$-locally private  $\Pi=((R_1,\dots,R_n),A)$ that performs Rademacher estimation, and any $m \leq n/8$, either
\begin{equation}
\label{eq:ber-inaccurate+}
\pr{}{\big| \Pi ( \Rad( \mu(m,n,\eps) ) ) - \mu(m,n,\eps) \big| \geq \frac{e^\eps +1}{e^\eps -1} \left(\frac{m}{2n} + \sqrt{\frac{\ln 6}{2n}} \right)} \geq \frac{1}{78}
\end{equation}
or there is an attack $M^{\Pi}_{m,n}$ such that
\begin{equation}
\label{eq:ber-inaccurate-half}
\pr{}{\left| \mathrm{Manip} \left(\Pi, \Rad(0), M^{\Pi}_{m,n} \right) \right| \geq \frac{e^\eps +1}{e^\eps -1} \left(\frac{m}{2n} + \sqrt{\frac{\ln 6}{2n}} \right) } \geq \frac{1}{78}
\end{equation}
That is, either the protocol is inaccurate absent a manipulation attack or inaccurate under some manipulation attack.
\end{thm}

\begin{proof}
By Corollary \ref{coro:reduction-to-rr-1} and Theorem \ref{thm:reduction-to-rr-2}, it is without loss of generality to assume $\Pi$ is an $\eps$-randomized response protocol. To reduce space, we will use shorthand $\alpha \gets \frac{e^\eps +1}{e^\eps -1} \left(\frac{m}{2n} + \sqrt{\frac{\ln 6}{2n}} \right)$ and $p \gets \mu(m,n,\eps)$.

We take the attack to be $M^\rr_{m,n}$. If \eqref{eq:ber-inaccurate-half} is true, then the proof is complete. Otherwise, we will prove \eqref{eq:ber-inaccurate+} under the premise that $\pr{}{\left| \mathrm{Manip} \left(\Pi, \Rad(0), M^{\rr}_{m,n} \right) \right| \geq \alpha}<1/78$. Observe that $p = 2\alpha$. 
\begin{align*}
\pr{}{|\Pi(\Rad(\mu))-p| < \alpha} &\leq \pr{}{|\Pi(\Rad(\mu))|\geq \alpha}\\
    &\leq 51 \cdot \pr{}{\left| \mathrm{Manip} \left(\Pi, \Rad(0), M^{\rr}_{m,n} \right) \right| \geq \alpha} + \frac{1}{3} \tag{Coro. \ref{coro:binary-attack}}\\
    &< \frac{77}{78}
\end{align*}
\eqref{eq:ber-inaccurate+} immediately follows from this upper bound. This concludes the proof.
\end{proof}

\subsection{Generalizing to Approximate Differential Privacy}

For clarity of exposition, we have limited the analysis to protocols that satisfy pure differential privacy. Here, we generalize our attack to approximate differential privacy.
\begin{equation}
\label{eq:rrr-delta}
\rrr_{\eps, \delta}(x) := \begin{cases}
    x & \mathrm{with~probability~} (1-\delta)\cdot \frac{e^\eps}{e^\eps + 1}\\
    -x & \mathrm{with~probability~} (1-\delta)\cdot \frac{1}{e^\eps + 1}\\
    2x & \mathrm{with~probability~} \delta
\end{cases}
\end{equation}
The randomized algorithm $\rrr_{\eps,\delta}$ ``fails at privacy'' with probability $\delta$: it reports an integer whose sign is the input $x$. Otherwise, it simply runs $\rrr_\eps$.

\begin{lem}[From \cite{KairouzOV15}]
\label{lem:reduction-to-rr-delta}
For any $(\eps,\delta)$-private randomizer $R: \{\pm 1\} \rightarrow \cY$, there exists a randomized algorithm $\redu{R}$ such that, for any $x \in \{\pm 1\}$, $\redu{R}(\rr_{\eps,\delta}(x))$ and $R(x)$.
\end{lem}

From this lemma, we may construct $(\eps,\delta)$ variants of $\redu{A}_{\vec{R},A}$ (Algorithm \ref{alg:transform-agg}) and $\redu{\Pi}$. Hence, we obtain these generalizations of Corollary \ref{coro:reduction-to-rr-1} and Theorem \ref{thm:reduction-to-rr-2}:

\begin{coro}
\label{coro:reduction-to-rr-delta-1}
For any $(\eps,\delta)$-locally private protocol $\Pi=((R_1,\dots, R_n), A)$ for binary data and any mean $\mu \in [-1,+1]$, $\Pi(\Rad(\mu))$ and $\redu{\Pi}(\Rad(\mu))$ are identically distributed.
\end{coro}

\begin{thm}
\label{thm:reduction-to-rr-delta-2}
Fix any $(\eps,\delta)$-locally private protocol $\Pi=((R_1,\dots, R_n), A)$ for binary data and any $m\leq n$. For any manipulation attack $\redu{M}$ against $\redu{\Pi}$, there exists an attack $M$ against $\Pi$ with the following property: for any mean $\mu \in [-1,+1]$, $\mathrm{Manip}_{m,n}(\Pi, \Rad(\mu), M)$ and $\mathrm{Manip}_{m,n}(\redu{\Pi}, \Rad(\mu), \redu{M})$ are identically distributed.
\end{thm}

We finally argue that $M^{\rr}_{m,n}$ is effective against any $(\eps,\delta)$ private randomized response protocol:

\begin{lem}
\label{lem:binary-attack-delta}
For any $n \geq 931$ and $m \leq n/8$, the distribution $\bigvec{\rrr_{\eps,\delta}}(\Rad(\mu(m,n,\eps)))$ cannot be distinguished from $\mathrm{Manip}\left(\bigvec{\rrr_{\eps,\delta}}, \Rad(0), M^{\rr}_{m,n} \right)$ with arbitrarily low probability of failure. Specifically, for all $Y \subseteq \{-2,-1,+1,+2\}^n$,
\[
\pr{}{\bigvec{\rrr_{\eps,\delta}} (\Rad(\mu(m,n,\eps))) \in Y} \leq 51 \cdot \pr{}{\mathrm{Manip}\left(\bigvec{\rrr_{\eps, \delta}}, \Rad(0), M^{\rr}_{m,n} \right) \in Y} + \frac{1}{3} + n\delta
\]
\end{lem}
\begin{proof}

For our proof, we use $\cF$ to denote the set of strings $F$ that contain at least one failure integer: $\cF := \{F \in \{-2,-1,+1,+2\}^n~|~ F \cap \{-2,+2\} \neq \emptyset \}$.

\begin{align*}
\pr{}{\bigvec{\rrr_{\eps,\delta}} (\Rad(\mu(m,n,\eps))^n) \in Y} =&~ \pr{}{\bigvec{\rrr_{\eps,\delta}} (\Rad(\mu(m,n,\eps))^n) \in Y - \cF}\\
    &+ \pr{}{\bigvec{\rrr_{\eps,\delta}} (\Rad(\mu(m,n,\eps))^n) \in Y \cap \cF}\\
    \leq&~ \pr{}{\bigvec{\rrr_{\eps,\delta}} (\Rad(\mu(m,n,\eps))^n) \in Y - \cF} + n\delta \tag{Union bound}\\
    \leq&~ 51 \cdot \pr{}{\mathrm{Manip}\left(\bigvec{\rrr_{\eps, \delta}}, \Rad(0), M^{\rr}_{m,n} \right) \in Y - \cF} + \frac{1}{3} + n\delta \tag{From Lemma \ref{lem:binary-attack}}\\
    \leq&~ 51 \cdot \pr{}{\mathrm{Manip}\left(\bigvec{\rrr_{\eps, \delta}}, \Rad(0), M^{\rr}_{m,n} \right) \in Y} + \frac{1}{3} + n\delta
\end{align*}

This concludes the proof.
\end{proof}
    
    \section{Attacks Against Protocols for Large Data Universes}
\label{sec:finite-attack}

In this section, we show that more powerful manipulation attacks are possible when the data universe is $[d]$ for $d>2$. For binary data, our attack showed that for any protocol there are two distributions $\bU$ and $\bP$ (i.e. $\Rad(0)$ and $\Rad(\mu(m,n,\eps)$) with large statistical distance but are indistinguishable under manipulation. Specifically, $\norm{\bU-\bP}_1 = \Omega ( \frac{1}{\eps \sqrt{n}} + \frac{m}{\eps n})$ where $\norm{\bU-\bP}_1$ denotes $\ell_1$ distance between the distributions $\sum_{j=1}^{d} |\bU(j)-\bP(j)|$. In this section, we show that there is an attack and a distribution such that $\norm{\bU-\bP}_1 = \Omega \big( \sqrt{\frac{d}{\log n}}(\frac{1}{\eps \sqrt{n}} + \frac{m}{\eps n}) \big)$ and $\bU,\bP$ are indistinguishable under this attack. This construction implies lower bounds for uniformity testing (given samples from $\bP$, determine if $\bP=\bU$ or if $\norm{\bP-\bU}_1$ is large) and $\ell_1$ estimation (given samples from $\bP$, report $\bP'$ such that $\norm{\bP-\bP'}_1$ is small).

Intuitively, our proof has the following structure. We show that, for every $\eps$ differentially private local randomizer $R:[d]\to \cY$, there is a set $H \subset [d]$ such that $R(\bU)$ and $R(\bU_H)$ are within an $\approx \eps / \sqrt{d}$ multiplicative factor of one another. The size of $H$ will be $d/2$ so $\norm{\bU - \bU_H}_1 = 1/2$. If a protocol $\Pi$ could distinguish between $\bU$ and $\bU_H$ then we would be able to create a protocol $\Pi'$ for binary data that distinguishes $\Rad(0)$ from $\Rad(1)$. Specifically, if $x_i=1$ then replace it with $x'_i \sim \bU_H$ and otherwise $x'_i \sim \bU_{\overline{H}}$ then run $\Pi$ on $x'_1,\dots,x'_n$. Since $\Pi'$ is $\eps/\sqrt{d}$ private, there must be a manipulation attack that defeats it when $m \approx n \eps / \sqrt{d}$. Our proof formalizes this intuition and generalizes it to the full range of $m$.

\subsection{A Family of Data Distributions}
\label{sec:distribution-family}

In this section, we show a particular way to convert a Rademacher distribution into a distribution over $[d]$. For a given partition of $[d]$ into $H,\overline{H}$ where $|H|=d/2$, we map the value $+1$ to a uniform element of $H$ and $-1$ to a uniform element of $\overline{H}$. Thus, when $x \sim \Rad(\mu)$, we obtain a corresponding random variable $\hat{x}$ over $[d]$ whose distribution is $\bP_{H,\mu}$ (see \eqref{eq:binary-to-finite} below). Notice that estimating $\pr{}{\hat{x} \in H}$ implies estimating $\mu$.

\begin{equation}
\label{eq:binary-to-finite}
    \bP_{H,\mu} := \begin{cases}
    \mathrm{Uniform~over}~H & \mathrm{with~probability} ~\half+\frac{\mu}{2}\\
    \mathrm{Uniform~over}~\overline{H} & \mathrm{otherwise}
    \end{cases}
\end{equation}

The algorithm $Q_{H,R}$ (Algorithm \ref{alg:randomizer-reduction}) performs the encoding of binary data $x\in \{\pm 1\}$ into $\hat{x} \in [d]$ then executes the randomizer $R$. Claim \ref{clm:randomizer-reduction} is immediate from the construction.

\begin{algorithm}
\caption{$Q_{H,R}$ a local randomizer for binary data}
\label{alg:randomizer-reduction}

\Parameters{A subset $H \subset [d]$ with size $d/2$; a local randomizer $R:[d] \to \cY$}

\KwIn{$x \in \{\pm 1\}$}
\KwOut{$y \in \cY$}

\medskip

If $x=1$ then sample $\hat{x}$ uniformly from $H$; otherwise, sample $\hat{x}$ uniformly from $\overline{H}$.

\Return{$y \sim R(\hat{x})$}

\end{algorithm}

\begin{clm}
\label{clm:randomizer-reduction}
For any local randomizer $R:[d] \to \cY$, $H\subset [d]$ with size $d/2$, and $p \in [-1,+1]$, the execution of $Q_{H,R}$ (Algorithm \ref{alg:randomizer-reduction}) on a value drawn from $\Rad(\mu)$ is equivalent with the execution of $R$ on a value drawn from $\bP_{H,\mu}$:
\[
Q_{H,R}(\Rad(\mu)) = R(\bP_{H,\mu})
\]
\end{clm}

So for every choice of $H$, $\mu$, and $\Pi = ((R_1,\dots,R_n),A)$, we can express each message from honest user $i$ as the output of a randomizer $Q_{H,R_i}$ that takes binary input. When these randomizers obey approximate differential privacy, Lemma \ref{lem:reduction-to-rr-delta} tells us that each randomizer can be decomposed into two algorithms, the first being $\rrr_{\eps,\delta}$ (see \eqref{eq:rrr-delta}). The second is a randomizer-dependent algorithm $\redu{Q_{H,R_i}}$. For brevity, we will use $\redu{\vec{Q}}_H$ to denote the vector of all $n$ of them.

\begin{lem}
\label{lem:finite-reduction-1}
Fix any protocol $\Pi = ((R_1,\dots,R_n),A)$ with data universe $[d]$, and any $H\subset[d]$ with size $d/2$. If each $Q_{H,R_i}$ satisfies $(\eps,\delta)$ privacy, then for any value $p \in [-1,+1]$,
\[
\vec{R}(\bP_{H,\mu}) = \redu{\vec{Q}}_H \left( \vecrrr_{\eps,\delta}(\Rad(\mu)) \right)
\]
\end{lem}

\subsection{The Attack}
In this subsection, we describe how $m$ corrupted users can attack an arbitrary $\eps$-private protocol $\Pi=((R_1,\dots,R_n),A)$. This attack, denoted $M^\Pi_{m,n}$, is sketched in Figure \ref{fig:finite-attack}. The first step is to sample a uniformly random $H$. We show that if this $H$ has the property that all $Q_{H,R_1}, \dots, Q_{H,R_n}$ all satisfy $(\eps',\delta)$ differential privacy, then this attack inherits guarantees from the attack $M^\rr_{m,n}$ against $\rr_{\eps',\delta}$. Then we show that this property holds with constant probability.

We begin the analysis of $M^\Pi_{m,n}$ by considering its behavior \emph{conditioned on a fixed choice of} $H$. This restricted form will be denoted $M^{\Pi,H}_{m,n}$. Then we will show that the random choice of $H$ gives the desired lower bound.

\begin{figure}    
    \begin{framed}
        Choose $H$ by uniformly sampling from all subsets of $[d]$ with size $d/2$
        
        Choose $C$, the users to corrupt, by uniformly sampling from all subsets of $[n]$ with size $m$
        
        Command each corrupted user $i\in S$ to report $y_i \sim \redu{Q}_{H, R_i}(+1)$
    \end{framed}
    
    \caption{$M^{\Pi}_{m,n}$, an attack against protocol $\Pi=(\vec{R},A)$ for $n$ users and data universe $[d]$}
    \label{fig:finite-attack}
\end{figure}

\subsubsection{Analysis for fixed set $H$}
Here, we show that manipulating $\Pi$ with $M^{\Pi,H}_{m,n}$ induces the same distribution as if we had manipulated randomized response $\rr_\eps$ with $M^{\rr}_{m,n}$:

\begin{clm}
\label{clm:finite-reduction-2}
Fix any protocol $\Pi = ((R_1, \dots, R_n), A)$ for data universe $[d]$, any $m \leq n$, and any $H \subset [d]$ with size $d/2$. If each $Q_{H,R_i}$ satisfies $(\eps,\delta)$ privacy, then for any value $p\in[-1,+1]$, the distribution  $\mathrm{Manip}\left(\vec{R}, \bP_{H,\mu}, M^{\Pi,H}_{m,n} \right) $ is identical to $\redu{\vec{Q}}_H \left( \mathrm{Manip}\left(\vecrrr_{\eps,\delta}, \Rad(\mu), M^{\rr}_{m,n} \right) \right)$
\end{clm}

\begin{proof}
The choice of $C$ in $M^{\rr}_{m,n}$ has the same distribution as in $M^{\Pi,H}_{m,n}$. To simplify the presentation, we assume that the users are sorted so that the corrupted set $C=[m]$.
\begin{align*}
\mathrm{Manip}\left(\vec{R}, \bP_{H,\mu}, M^{\Pi,H}_{m,n} \right) &= (\redu{Q_{H,R_i}}(1))_{i \leq m} \times (R_i(\bP_{H,\mu}))_{i > m} \tag{By construction}\\
    &= (\redu{Q_{H,R_i}}(1))_{i \leq m} \times \left( Q_{H,R_i} (\Rad(\mu) ) \right)_{i > m} \tag{Claim \ref{clm:randomizer-reduction}}\\
    &= (\redu{Q_{H,R_i}}(1))_{i \leq m} \times \left( \redu{Q_{H,R_i}} \left(\rrr_{\eps,\delta}\left(\Rad(\mu) \right) \right) \right)_{i > m} \tag{Lemma \ref{lem:reduction-to-rr-delta}}\\
    &= \redu{\vec{Q}}_H \left( \mathrm{Manip}\left(\vecrrr_{\eps,\delta}, \Rad(\mu), M^{\rr}_{m,n} \right) \right)
\end{align*}
This concludes the proof.
\end{proof}

Lemma \ref{lem:finite-reduction-1} and Claim \ref{clm:finite-reduction-2} imply that we can use the analysis of $M^{\rr}_{m, n}$ for our new attack $M^{\Pi, H}_{m, n}$ \emph{provided that} $(\eps,\delta)$ privacy holds for all $Q_{H,R_i}$:
\begin{lem}
\label{lem:if-then-finite-attack}
Fix any protocol $\Pi = ((R_1, \dots, R_n), A)$ for data universe $[d]$, any $m\leq n$, and any $H \subset [d]$ with size $d/2$. There exists a value $\mu \in [-1,+1]$ such that
\begin{equation}
\label{eq:finite-distance}
\norm{\bU-\bP_{H,\mu}}_1 = \frac{e^{\eps}  +1}{e^{\eps} - 1} \cdot \left( \frac{m}{n} + \sqrt{\frac{2\ln 6}{n}} \right)
\end{equation}
but if each $Q_{H,R_i}$ is $(\eps, \delta)$ differentially private, then for any $Y \subset \cY^n$,
\begin{equation}
\label{eq:finite-similarity}
\pr{}{\vec{R}(\bP_{H,\mu}) \in Y} \leq 51 \cdot \pr{}{ \mathrm{Manip}\left(\vec{R}, \bU, M^{\Pi,H}_{m,n} \right) \in Y} + \frac{1}{3} + n\delta
\end{equation}
\end{lem}

\begin{proof}
Recall the function $\mu(m,n,\eps) = \frac{e^\eps + 1}{e^\eps - 1} \big( \frac{m}{n} + \sqrt{\frac{2\ln 6}{n}} \big)$. We will set $\mu \gets \mu(m,n,\eps)$. By Lemma \ref{lem:binary-attack-delta}, we have
\[
\pr{}{\redu{\vec{Q}}_H\left( \vecrrr_{\eps,\delta}( \Rad(\mu) ) \right) \in Y} \leq 51 \cdot \pr{}{ \redu{\vec{Q}}_H \left( \mathrm{Manip} \left( \vecrrr_{\eps,\delta}, \Rad(0); M^{\rr}_{m,n} \right) \right) \in Y} + \frac{1}{3} + n\delta
\]
and by Lemma \ref{lem:finite-reduction-1} and Claim \ref{clm:finite-reduction-2} that statement implies \eqref{eq:finite-similarity}.

It remains to prove \eqref{eq:finite-distance}. When sampling $x \sim \bP_{H,\mu}$, the probability that $x = h$ is $\frac{1+\mu}{d}$ for each $h \in H$ and $\frac{1-\mu}{d}$ for each $h \notin H$. Hence,
\begin{align*}
\norm{\bU - \bP_{H,\mu}}_1 &= \frac{d}{2}\cdot \left|\frac{1}{d} - \frac{1+\mu}{d} \right| + \frac{d}{2}\cdot \left| \frac{1}{d} - \frac{1-\mu}{d} \right|\\
    &= \mu \\
    &= \frac{e^{\eps} + 1}{e^{\eps} - 1} \cdot \left( \frac{m}{n} + \sqrt{\frac{2\ln 6}{n}} \right)
\end{align*}
This concludes the proof.
\end{proof}

\subsubsection{Analysis for randomized $H$}

Here, we obtain a lower bound by analyzing randomness in $H$. We begin with a lemma that bounds the privacy parameters of all $Q_{H,R_i}$ by an $\eps'$ that depends on the structure of $\vec{R}$: we will use $|\vec{R}|_{\neq}$ to denote the number of unique randomizers in $\vec{R}$.

\begin{lem}
\label{lem:amplified-privacy-1}
Fix any $\vec{R}$ where each $R_i : [d] \rightarrow \cY$ is $\eps$ differentially private. There is a constant $c$ such that, if $d > c\cdot (e^{2\eps}-1)^2 \ln \left(|\cY| \cdot |\vec{R}|_{\neq} \right)$ and $H$ is drawn uniformly from all subsets of $[d]$ with size $d/2$, then the following holds with probability $> 2/3$ over the randomness of $H$: Every $Q_{H, R_i}$ specified by Algorithm \ref{alg:randomizer-reduction} is $\eps'$ differentially private, where
\[
\eps' = (e^{2\eps}-1) \sqrt{\frac{c}{d} \ln \left( |\cY| \cdot |\vec{R}|_{\neq} \right)}
\]
\end{lem}

We continue with a bound that only depends on $n$ and not any particular structure in $\vec{R}$:

\begin{lem}
\label{lem:amplified-privacy-2}
Fix any $\vec{R}$ where each $R_i : [d] \rightarrow \cY$ is $\eps$ differentially private. There is a constant $c$ such that, if $d > c\cdot (e^{2\eps}-1)^2 \ln \left(e^\eps n \right)$ and $H$ is drawn uniformly from all subsets of $[d]$ with size $d/2$, then the following holds with probability $> 2/3$ over the randomness of $H$: Every $Q_{H, R_i}$ specified by Algorithm \ref{alg:randomizer-reduction} is $(\eps', 1/180n)$ differentially private, where
\[
\eps' = (e^{2\eps}-1) \sqrt{\frac{c}{d} \ln \left(e^\eps n \right)}
\]
\end{lem}

Proofs of these statements can be found in Appendix \ref{apdx:finite-attack}. From Lemmas \ref{lem:if-then-finite-attack}, \ref{lem:amplified-privacy-1}, and \ref{lem:amplified-privacy-2}, we have an attack that, with probability $2/3$, successfully obscures a uniform distribution:

\begin{lem}
\label{lem:finite-attack}
Fix any $n \geq 931$, any $m \leq n/8$, any $\eps < 1$, and any $\eps$-locally private protocol $\Pi=(\vec{R},A)$ that accepts data from $[d]$. There are constants $c_0,c_1$ and a value $p \in [-1,+1]$ such that, for all $H\subset [d]$ with size $|H|=d/2$,
\[
\norm{\bU-\bP_{H,\mu}}_1 \geq \frac{c_1 \cdot \sqrt{d}}{\eps \sqrt{ \ln \left(\min n, |\cY| \cdot |\vec{R}|_{\neq} \right)}} \cdot \left( \frac{m}{n} + \sqrt{\frac{1}{n}} \right)
\]
but if $d > c_0\cdot (e^{2\eps}-1)^2 \ln \left(\min n, |\cY| \cdot |\vec{R}|_{\neq} \right)$, then the following holds with probability $> 2/3$ over the random choice of $H$ in $M^{\Pi}_{m,n}$ (Figure \ref{fig:finite-attack})
\[
\pr{}{ \vec{R}( \bP_{H,\mu} ) \in Y} \leq 51 \cdot \pr{}{ \mathrm{Manip} \left( \vec{R}, \bU, M^{\Pi,H}_{m,n} \right) \in Y } + \frac{61}{180}
\]
\end{lem}

\subsection{Applications to Testing and Estimation}
From Lemma \ref{lem:finite-attack}, we immediately derive a lower bound on how well the manipulation attack fares against uniformity testers:
\begin{thm}
\label{thm:uniformity-attack}
Fix any $n \geq 931$, any $m \leq n/8$, any $\eps < 1$, and any $\eps$-locally private protocol $\Pi=(\vec{R},A)$ for testing uniformity over $[d]$. There are constants $c_0,c_1$ such that for all $d > c_0\cdot (e^{2\eps}-1)^2 \ln \left(\min n, |\cY| \cdot |\vec{R}|_{\neq} \right)$ and all distributions $\bP$ that satisfy
\begin{equation}
\label{eq:uniformity-separation}
\norm{\bU-\bP}_1 \geq \frac{c_1 \cdot \sqrt{d}}{\eps \sqrt{ \ln \left(\min n, |\cY| \cdot |\vec{R}|_{\neq} \right)}} \cdot \left( \frac{m}{n} + \sqrt{\frac{1}{n}} \right)
\end{equation}
at least one of the following holds:
\begin{align*}
\pr{}{ \Pi( \bP ) = \mathrm{``uniform"}} &\geq \frac{1}{80}\\
\pr{}{ \mathrm{Manip} \left(\Pi, \bU, M^{\Pi}_{m,n} \right) = \mathrm{``not~uniform"} } &\geq \frac{1}{120}
\end{align*}
\end{thm}

\medskip

Directly applying Lemma \ref{lem:finite-attack} to distribution estimation would give an $\Omega\left(\sqrt{\frac{d}{\log n}} \left( \frac{m}{\eps n} + \frac{1}{\eps \sqrt{n}} \right) \right)$ lower bound on error in $\ell_1$ distance. But a theorem in \cite{YeB18} implies that estimation protocols must have error $\Omega \left( \frac{d}{\eps \sqrt{n}} \right)$ in $\ell_1$ distance. We integrate these two results below:
\begin{thm}
\label{thm:estimation-attack}
Fix any $n \geq 931$, any $m\leq n/8$, and any $\eps \in [-1,+1]$. There exists constants $c_0,c_1$ such that, for any $\eps$-locally private protocol $\Pi=(\vec{R},A)$ that estimates distributions over $[d]$ where $d > c_0\cdot (e^{2\eps}-1)^2 \ln \left(\min n, |\cY| \cdot |\vec{R}|_{\neq} \right)$, there exists a distribution $\bP$ where at least one of the following holds:
\begin{align*}
\pr{}{\norm{ \Pi ( \bP ) - \bP }_1 \geq \frac{c_1}{\eps} \left( \frac{d}{\sqrt{n}} + \frac{m}{n} \cdot \sqrt{\frac{d }{\ln \left(\min n, |\cY| \cdot |\vec{R}|_{\neq} \right)}} \right) } &> \frac{1}{80}\\
\pr{}{\norm{\mathrm{Manip} \left(\Pi, \bU, M^{\Pi}_{m,n} \right) - \bU }_1 \geq \frac{c_1}{\eps} \left( \frac{d}{\sqrt{n}} + \frac{m}{n} \cdot \sqrt{\frac{d }{\ln \left(\min n, |\cY| \cdot |\vec{R}|_{\neq} \right)}} \right) } &> \frac{1}{120}
\end{align*}
\end{thm}
    
    
    \section{Protocols with Nearly Optimal Robustness to Manipulation}
\label{sec:protocols}
In this section, we consider a number of well-studied problems in local privacy and identify specific protocols from the literature that have optimal robustness to manipulation (i.e. matching the lower bounds implied by our attacks). As discussed in the introduction, most of these problems can be cast as accurately estimating the mean of bounded vectors.


\subsection{Warmup: Mean Estimation for Binary Data}
\label{sec:binary-protocol}

As a warmup, we analyze the randomized response protocol in the presence of manipulation.  Recall that the protocol is defined by the local randomized $\rrr_\eps$ and aggregator $\rra_{n,\eps}$ as follows (where we have rescaled the messages to be an unbiased estimate of $x$, which is more convenient for analysis):

\begin{align*}
\rrr_\eps(x) &:= \begin{cases} \frac{e^\eps + 1}{e^\eps - 1} \cdot x &\mathrm{with~probability}~\frac{e^\eps}{e^\eps + 1}\\
    -\frac{e^\eps + 1}{e^\eps - 1} \cdot x &\mathrm{with~probability}~\frac{1}{e^\eps + 1}
    \end{cases} \\
\rra_{n,\eps}(\vec{y}) &:= \frac{1}{n} \sum_{i=1}^n y_i \end{align*}

We bound the error of this protocol by $O(\frac{1}{\eps}(\frac{1}{\sqrt{n}}+\frac{m}{n})$, which matches the lower bound of Theorem \ref{thm:binary-attack} up to constants.

\begin{thm}
\label{thm:rr-under-attack}
For any positive integers $m \leq n$, any $\eps > 0$, any $\vec{x} \in \zo^{n}$, any manipulation adversary $M$, and any $\beta > 0$,
\[
\pr{}{\left|\mathrm{Manip}_{m,n}(\rr_{\eps, n}, \vec{x}, M) - \frac{1}{n}\sum_{i=1}^n x_i \right| < \frac{e^\eps + 1}{e^\eps - 1} \cdot \left( \sqrt{\frac{2}{n} \ln \frac{2}{\beta}} +  \frac{2m}{n}  \right) } \geq 1-\beta
\]
\end{thm}

\begin{proof}
Consider an execution of $\mathrm{Manip}(\rr_{\eps, n}, \vec{x}, M)$.  Let $C$ be the set of users corrupted by $M$, let $y_1,\dots,y_n$ be the messages sent in the protocol and let $\vec{\underline{y}}$ be the messages that would have been sent in an honest execution (so $\underline{y}_i = y_i$ for every $i \not\in C$).  Let
$
z = \frac{1}{n} \sum_{i=1}^n y_i
$
be the output of the aggregator.

We can break up the error into two components, one corresponding to the error of the honest execution and one corresponding to the error introduced by manipulation.
\begin{align*}
\left|\frac{1}{n} \sum_{i \in [n]} y_i - \frac{1}{n}\sum_{i \in [n]} x_i \right|
={} &\left|\frac{1}{n} \sum_{i \in [n]} y_i - \frac{1}{n} \sum_{i \in [n]} \underline{y}_i + \frac{1}{n} \sum_{i \in [n]} \underline{y}_i - \frac{1}{n}\sum_{i \in [n]} x_i \right| \\
\leq{} &\left| \frac{1}{n} \sum_{i \in [n]} y_i - \frac{1}{n} \sum_{i \in [n]} \underline{y}_i \right| + \left| \frac{1}{n} \sum_{i \in [n]} \underline{y}_i - \frac{1}{n}\sum_{i \in [n]} x_i \right| \\
={} &\underbrace{\left| \frac{1}{n} \sum_{i \in C} y_i - \underline{y}_i \right|}_{\textrm{manipulation}} + \underbrace{\left| \frac{1}{n} \sum_{i \in [n]} \underline{y}_i - \frac{1}{n}\sum_{i \in [n]} x_i \right|}_{\textrm{honest execution}}
\end{align*}
Since each message in the protocol is either $\frac{e^{\eps}+1}{e^{\eps}-1}$ or $-\frac{e^\eps + 1}{e^{\eps}-1}$, we have $|y_i - \underline{y}_i| \leq 2\cdot \frac{e^\eps + 1}{e^\eps - 1}$. Thus, the manipulation term is bounded by $\frac{e^\eps + 1}{e^\eps - 1} \cdot \frac{2m}{n}$.  

For the error of the honest execution, note that $\ex{}{\underline{y}_i} = x_i$ and $\frac{1}{n} \sum_{i \in [n]} \underline{y}_i$ is an average of $n$ independent random variables bounded to a range of width $2\cdot \frac{e^\eps + 1}{e^\eps - 1}$.  Thus, by Hoeffding's inequality, we have that with probability at least $1-\beta$, the second term is bounded by $\frac{e^{\eps}+1}{e^{\eps} - 1}\sqrt{\frac{2 \ln(2/\beta)}{n}}$ with probability at least $1-\beta$.
\end{proof}

Our analysis of richer protocols has the same structure.  We construct the protocol so that each message $y_i$ gives an unbiased estimate of $x_i$, and the aggregation computes the mean of the messages.  We then isolate the effect of the manipulation from that of an honest execution.  Finally, we have to bound the degree to which a set of $m$ messages influences the output of the protocol.  For richer protocols the analysis of the final step will become more involved.

\subsection{Mean Estimation}
We consider vector-valued data in $\R^d$. For any $p\geq 1$, $\norm{x}_p := (\sum_{j=1}^d |x_j|^p)^{1/p}$ denotes the standard $\ell_p$ norm and $B^d_p$ denotes the $\ell_p$ unit ball in $\R^d$. As is standard $\norm{x}_\infty = \max_{j\in[d]} |x_j|$ is the $\ell_\infty$ norm and $B^d_\infty$ is the $\ell_\infty$ unit ball. In this section, we study instances of the general $\ell_p / \ell_q$ \emph{mean estimation problem}: given data $x_1,\dots,x_n \in B^d_p$, output some $\hat{\mu}$ such that $\norm{\hat{\mu} - \frac{1}{n} \sum_i x_i}_q$ is as small as possible. 


\subsubsection{$\ell_\infty / \ell_\infty$ estimation (Counting Queries)}

In this problem, each user has data $x_i\in B^d_\infty$ and the goal is to obtain a vector $\hat{\mu}$ such that $\norm{\hat{\mu}-\frac{1}{n} \sum x_i}_\infty$ is as small as possible. We consider the following protocol $\esti = (\estir, n, \estia)$, which is known to have optimal error absent manipulation.

\begin{enumerate}
    \item Using public randomness, we partition users into $d$ groups each of size $n/d$. Intuitively, we are assigning each group to one coordinate. 
    \item For each group $j$, each user $i$ in group $j$ reports the message $y_i\gets \rrr(x_{i,j})$ 
    \item For each group $j$, the aggregator computes the average of the messages from group $j$ to obtain $\hat{\mu}_j \approx \frac{1}{n} \sum_i x_{i,j}$. The aggregator reports $\hat{\mu} = (\hat{\mu}_1, \dots, \hat{\mu}_d)$
\end{enumerate}

If the adversary's corruptions are oblivious to the public partition, then we show that there are $\approx m/d$ corrupt users in each group of size $n/d$. By our analysis of randomized response, the adversary can introduce at most $\approx \frac{m/d}{\eps n/d} = \frac{m}{\eps n}$ error in any single coordinate.

\begin{thm}
\label{thm:esti-under-attack}
For any $\eps \in (0,1)$, any positive integers $m \leq n$, any $x_1,\dots,x_n \in B^d_\infty$, and any \emph{public-string-oblivious} adversary $M$, with probability $\geq 99/100$, we have
\[
\norm{\mathrm{Manip}_{m,n}(\esti_{\eps}, \vec{x}, M) - \frac{1}{n}\sum_{i=1}^n x_i }_\infty = O\left( \sqrt{\frac{d\log d}{\eps^2 n} } + \frac{m}{\eps n} \right) 
\]
\end{thm}

Observe that the dependence on $m$ matches that of the lower bound in Theorem \ref{thm:binary-attack} for Bernoulli estimation. We give the full details of the protocol in Appendix \ref{apdx:esti}. 

\subsubsection{$\ell_1/\ell_\infty$ Estimation (Histograms)}
\label{sec:histogram}

In this problem, each user $i$ has data $x_i \in B^d_1$ and the objective is a $\hat{\mu}$ such that $\norm{\hat{\mu} - \frac{1}{n}\sum_{i=1}^n x_i}_\infty$ is as small as possible. To simplify the discussion, we focus on the special case where user $i$ has data $x_i \in [d]$. Define $\mathit{freq}(j,\vec{x}) := \frac{1}{n}\sum_{i=1}^n \indic{x_i = j}$ and $\mathit{freq}(\vec{x}) := (\mathit{freq}(1,\vec{x}), \dots, \mathit{freq}(1,\vec{x}))$. The objective is a vector $\hat{\mu}$ such that $\norm{\hat{\mu}-\mathit{freq}(\vec{x})}_\infty$ is as small as possible.

We consider the following protocol $\hst$,\footnote{In \cite{BassilyNST17} the protocol is called \texttt{ExplicitHist}.}, which is known to have optimal error absent manipulation:

\begin{enumerate}
    \item For each user $i$, independently sample a uniform public vector $\vec{s}_i \in \{\pm 1\}^{d}$.
    \item Each user $i$ reports the message $y_i \gets \rrr_\eps(s_{i,x_i})$ to the aggregator.
    \item The aggregator receives messages $y_1,\dots,y_n$ and outputs $\hat{\mu} \gets \frac{1}{n} \sum_{i=1}^n y_i \cdot \vec{s}_i$.
\end{enumerate}


\begin{thm}
\label{thm:hst-linfty}
For any $\eps \in (0,1)$, any positive integers $m \leq n$, any $x_1,\dots,x_n\in [d]$, and any adversary $M$, with probability $\geq 99/100$, we have
\[
\norm{\mathrm{Manip}_{m,n}(\hst_\eps,\vec{x},M) - \mathit{freq}(\vec{x}) }_\infty = O \left(\sqrt{\frac{\log d}{\eps^2 n}} + \frac{m}{\eps n}\right)
\]
\end{thm}

\begin{proof}[Proof Sketch.]
Identically to the proof of Theorem \ref{thm:rr-under-attack}, we partition the error contributed by the honest and corrupt users. Let $\vec{y}$ be the messages sent in the protocol and let $\vec{\underline{y}}$ be the messages that would have been sent in an honest execution.
\begin{align*}
\norm{\frac{1}{n} \sum_{i=1}^n y_i \vec{s}_i - \mathit{freq}(\vec{x}) }_\infty &= \norm{\frac{1}{n} \sum_{i=1}^n y_i \vec{s}_i - \frac{1}{n}\sum_{i=1}^n \underline{y}_i \vec{s}_i + \frac{1}{n}\sum_{i=1}^n \underline{y}_i \vec{s}_i - \mathit{freq}(\vec{x}) }_\infty\\
    &\leq \norm{ \frac{1}{n} \sum_{i=1}^n y_i \vec{s}_i - \frac{1}{n}\sum_{i=1}^n \underline{y}_i \vec{s}_i }_\infty + \norm{ \frac{1}{n}\sum_{i=1}^n \underline{y}_i \vec{s}_i - \mathit{freq}(\vec{x}) }_\infty \\
    &= \underbrace{\max_{j \in [d]} \left[ \frac{1}{n} \sum_{i\in C} (y_i - \underline{y}_i) s_{i,j} \right]}_{\textrm{manipulation}} + \underbrace{\max_{j \in [d]} \left| \frac{1}{n} \sum_{i=1}^n \underline{y}_i s_{i,j} - \indic{x_i=j} \right|}_{\textrm{honest execution}}    
\end{align*}

To bound the error from the manipulation, note that messages have magnitude $\frac{e^\eps+1}{e^\eps - 1} = \Theta(1/\eps)$. Hence, the bias introduced to any coordinate $j$ is at most $O(m/\eps n)$ with probability 1.

sWe now bound the error introduced by the honest execution of the protocol. If $x_i = j$, the expectation of $\underline{y}_i s_{i,j}$ is 1. Otherwise, the expectation is 0 because of pairwise independence. Hence, the honest execution has 0 expected error. Because messages have magnitude $\Theta(1/\eps)$, Hoeffding's inequality and a union bound imply that no frequency estimate is more than $O(\sqrt{\log d / \eps^2 n})$ from $\mathit{freq}(j,\vec{x})$ with probability $\geq 99/100$. This concludes the proof.
\end{proof}

\medskip

A slightly more general protocol can be used to obtain the same result for $\ell_1/\ell_\infty$ estimation.
\begin{thm}
\label{thm:esto-infty-informal}
For any $\eps \in (0,1)$, there is an $\eps$-locally private protocol $\esto_\eps$ such that for any positive integer $n$, any $x_1,\dots,x_n \in B^d_1$, and any adversary $M$, with probability $\geq 99/100$, we have
\[
\norm{\mathrm{Manip}_{m,n}(\esto_\eps,\vec{x},M) - \frac{1}{n}\sum_{i=1}^n x_i}_\infty = O\left( \sqrt{\frac{\log d}{\eps^2 n}} + \frac{m}{\eps n} \right)
\]
\end{thm}
Observe that the manipulation error matches that of the lower bound in Theorem \ref{thm:binary-attack} for Bernoulli estimation. We give the full details of the protocol $\esto_\eps$ in Appendix \ref{apdx:esto}.

\subsubsection{$\ell_1/\ell_1$ Estimation (Frequency Estimation)}
\label{sec:estimation-protocol}
In this problem, each user $i$ has data $x_i \in B^d_1$ and the objective is a $\hat{\mu}$ such that $\norm{\hat{\mu} - \frac{1}{n}\sum_{i=1}^n x_i}_1$ is as small as possible. Because this problem and the $\ell_1/\ell_\infty$ problem have the same data type, we consider the same protocols but change the analysis to upper bound $\ell_1$ error.

\begin{thm}
\label{thm:hst-l1-informal}
For any $\eps \in (0,1)$, any positive integer $n$, any $x_1,\dots,x_n \in [d]$, and any adversary $M$, with probability $\geq 99/100$, we have
\[
\norm{\mathrm{Manip}_{m,n}(\hst_\eps, \vec{x},M) - \frac{1}{n}\sum_{i=1}^n x_i }_1 =  O \left( \sqrt{ \frac{d^2\log n}{\eps^2 n} } + \frac{m\sqrt{d\log n} }{\eps n} \right)
\]
\end{thm}

\begin{proof}[Proof Sketch.]
Identically to the proof of Theorem \ref{thm:rr-under-attack}, we partition the error contributed by the honest and corrupt users. Let $\vec{y}$ be the messages sent in the protocol and let $\vec{\underline{y}}$ be the messages that would have been sent in an honest execution.  Let $S \in \pmo^{d \times n}$ be the matrix whose columns are $\vec{s}_1,\dots,\vec{s}_n$, and $S_{C} \in \pmo^{d \times |C|}$ be the submatrix consisting only columns corresponding to users $i \in C$.
\begin{align*}
\norm{\frac{1}{n} \sum_{i=1}^n y_i \vec{s}_i - \mathit{freq}(\vec{x}) }_1 &= \norm{\frac{1}{n} \sum_{i=1}^n y_i \vec{s}_i - \frac{1}{n}\sum_{i=1}^n \underline{y}_i \vec{s}_i + \frac{1}{n}\sum_{i=1}^n \underline{y}_i \vec{s}_i - \mathit{freq}(\vec{x}) }_1\\
    &\leq \norm{ \frac{1}{n} \sum_{i=1}^n y_i \vec{s}_i - \frac{1}{n}\sum_{i=1}^n \underline{y}_i \vec{s}_i }_1 + \norm{ \frac{1}{n}\sum_{i=1}^n \underline{y}_i \vec{s}_i - \mathit{freq}(\vec{x}) }_1 \\
    &= \sum_{j \in [d]} \left| \frac{1}{n} \sum_{i\in C} (y_i - \underline{y}_i) s_{i,j} \right| + \sum_{j \in [d]} \left| \frac{1}{n} \sum_{i=1}^n \underline{y}_i s_{i,j} - \indic{x_i=j} \right| \\
    &= \underbrace{ \left\| \frac{1}{n} S_{C}(\vec{y}_{C} - \vec{\underline{y}}_{C}) \right\|_{1} }_{\textrm{manipulation}} + \underbrace{\sum_{j \in [d]} \left| \frac{1}{n} \sum_{i=1}^n \underline{y}_i s_{i,j} - \indic{x_i=j} \right|}_{\textrm{honest execution}}
\end{align*}

To bound the error from the honest execution, observe that the expectation and variance are $O(\sqrt{1/\eps^2 n})$ and $O(1/\eps n)$, respectively, for any term in the outer sum. Hence, error has magnitude $O(\sqrt{d^2/\eps^2 n})$ with probability $\geq 199/200$.

To bound the error from the manipulation, we will use bounds on the singular values of the random matrix $S_C$.  As a shorthand, let $c_{\eps} = \frac{e^{\eps}+1}{e^{\eps}-1}$.  Then we have
\begin{align*}
    \left\| \frac{1}{n} S_{C}(\vec{y}_{C} - \vec{\underline{y}}_{C}) \right\|_{1}
    \leq{} & \frac{1}{n} \max_{C \subseteq [n]} \left\|S_{C}(\vec{y}_{C} - \vec{\underline{y}}_{C}) \right\|_{1} \\
    \leq{} & \frac{2}{n} \max_{C \subseteq [n] \atop |C| = m}~\max_{\vec{y}_{C} \in \{-c_\eps,c_\eps\}^m} \left\| S_{C} \vec{y}_{C} \right\|_1 \\
    ={} &\frac{2}{n} \max_{C \subseteq [n] \atop |C| = m}~\max_{\vec{y}_{C} \in \R^m \atop \|\vec{y}\|_{2} \leq c_{\eps} \sqrt{m}} \left\|  S_{C} \vec{y}_{C} \right\|_1 \\
    \leq{} &\frac{c_{\eps} \sqrt{m d}}{n} \max_{C \subseteq [n] \atop |C| = m}~\max_{\vec{y}_{C} \in \R^m \atop \|\vec{y}\|_{2} \leq 1} \left\|  S_{C} \vec{y}_{C} \right\|_2 \\
    ={} &\frac{c_{\eps}\sqrt{md}}{n} \max_{C \subseteq [n] \atop |C| = m} \| S_{C} \|_2
\end{align*}
where $\| S_{C} \|_2$ denotes the largest singular value (operator norm) of $S_{C}$.  Since each matrix $S_{C} \in \pmo^{d \times m}$ is uniformly random, we can use strong bounds on the singular values of random matrices.

\begin{lem}[see e.g.~the textbook \cite{Tao12}]
\label{lem:random-matrix-product}
For any $k \in \R_+$ larger than an absolute constant and a matrix $S_C \in \R^{d \times m}$ whose entries are sampled independently and identically, the following holds with probability $\geq 1- \exp(-k(d+m))$ over the randomness of $S_C$.
\[
\norm{S_C}_{2} = O( \sqrt{kd} + \sqrt{km} )
\]
\end{lem}
The adversary has $\binom{n}{m} \leq \exp(m\ln n)$ choices of corruptions $C$. By a union bound over that set, we have with probability $\geq 1- \exp(m\ln n - k(m+d))$
\begin{align*}
\left\| \frac{1}{n} S_{C}(\vec{y}_{C} - \vec{\underline{y}}_{C}) \right\|_{1} &\leq \frac{c_{\eps}\sqrt{md}}{n} \cdot O(\sqrt{kd} + \sqrt{km})\\
    &= O \left( \sqrt{ \frac{d^2k}{\eps^2 n} } + \frac{m\sqrt{dk}}{\eps n} \right)
\end{align*}
The probability is $\geq 199/200$ when $k=O(\ln n)$. A union bound over the manipulation and honest execution completes the proof.
\end{proof}

\medskip

A slightly more general protocol can be used to obtain the same result for $\ell_1/\ell_1$ estimation.
\begin{thm}
\label{thm:esto-l1-informal}
For any $\eps \in (0,1)$, there is an $\eps$-locally private protocol $\esto$ such that for any positive integer $n$, any $x_1,\dots,x_n \in (B^d_1)^n$, and any adversary $M$, with probability $\geq 99/100$, we have
\[
\norm{\mathrm{Manip}_{m,n}(\esto_\eps, \vec{x},M) - \frac{1}{n}\sum_{i=1}^n x_i }_1 = O \left( \sqrt{ \frac{d^2\log n}{\eps^2 n} } + \frac{m\sqrt{d\log n}}{\eps n} \right)
\]
\end{thm}

Observe that the manipulation error matches the lower bound in Theorem \ref{thm:estimation-attack}, up to a logarithmic factor.


\subsubsection{$\ell_2 / \ell_2$ Estimation}
In this problem, each user $i$ has data $x_i \in B^d_2$ and the objective is a $\hat{\mu}$ such that $\norm{\hat{\mu} - \frac{1}{n}\sum_{i=1}^n x_i}_2$ is as small as possible.

Consider the protocol $\estt$,\footnote{The protocol is a variation of one described in Section 4.2.3 of \cite{DuchiJW13}.} described below
\begin{enumerate}
    \item For each user $i$, we sample $\vec{s}_i \in \R^d$ uniformly at random from the surface of $B^d_2$.
    \item Each user $i$ computes $w_i \gets \mathrm{sgn}(\vec{s}_i \cdot x_i)$ and then reports $y_i \gets \rrr_\eps(w_i)$ to the aggregator
    \item The aggregator receives the messages $y_1,\dots,y_n$ and outputs $\vec{z} \gets \frac{c\sqrt{d}}{n} \sum_{i=1}^n y_i \vec{s}_i$ for some absolute constant $c$.
\end{enumerate}

\begin{thm}
\label{thm:estt-under-attack}
For any $\eps \in (0,1)$, any positive integer $n$, any $x_1,\dots,x_n \in B^d_2$, and any adversary $M$, with probability $\geq 99/100$, we have
\[
\norm{\mathrm{Manip}_{m,n}(\estt_\eps, \vec{x},M) - \frac{1}{n}\sum_{i=1}^n x_i }_2 = O \left( \sqrt{ \frac{d \log n }{\eps^2 n} } + \frac{m\sqrt{\log n}}{\eps n} \right)
\]
\end{thm}

\begin{proof}[Proof Sketch]
Identically to the proof of Theorem \ref{thm:rr-under-attack}, we partition the error contributed by the honest and corrupt users.  Let $S \in \pmo^{d \times n}$ be the matrix whose columns are $\vec{s}_1,\dots,\vec{s}_n$, and $S_{C} \in \pmo^{d \times |C|}$ be the submatrix consisting only columns corresponding to users $i \in C$.
\begin{align*}
\norm{\frac{c\sqrt{d}}{n} \sum_{i=1}^n y_i \vec{s}_i - \frac{1}{n}\sum_{i =1}^n x_i }_2 &= \norm{\frac{c\sqrt{d}}{n} \sum_{i=1}^n y_i \vec{s}_i - \frac{c\sqrt{d}}{n}\sum_{i=1}^n \underline{y}_i \vec{s}_i + \frac{c\sqrt{d}}{n}\sum_{i=1}^n \underline{y}_i \vec{s}_i - \frac{1}{n}\sum_{i =1}^n x_i }_2\\
    &\leq \norm{ \frac{c\sqrt{d}}{n} \sum_{i=1}^n y_i \vec{s}_i - \frac{c\sqrt{d}}{n}\sum_{i=1}^n \underline{y}_i \vec{s}_i }_2 + \norm{ \frac{c\sqrt{d}}{n}\sum_{i=1}^n \underline{y}_i \vec{s}_i - \frac{1}{n}\sum_{i =1}^n x_i }_2 \\
    &= \norm{ \frac{c\sqrt{d}}{n} \sum_{i\in C} (y_i - \underline{y}_i) s_{i,j} }_2 + \norm{ \frac{c\sqrt{d}}{n}\sum_{i=1}^n \underline{y}_i \vec{s}_i - \frac{1}{n}\sum_{i =1}^n x_i }_2\\
    &= \underbrace{ \norm{ \frac{c\sqrt{d}}{n} S_C(\vec{y}_C-\vec{\underline{y}}_C)}_2 }_{\textrm{manipulation}} + \underbrace{\norm{ \frac{c\sqrt{d}}{n}\sum_{i=1}^n \underline{y}_i \vec{s}_i - \frac{1}{n}\sum_{i =1}^n x_i }_2}_{\textrm{honest execution}}
\end{align*}

A lemma from \cite{DuchiJW13} implies that the error introduced by the honest execution of the protocol is $O(\sqrt{d/\eps^2 n})$ with probability $\geq 299/300$.

To bound the error from the manipulation,  we will again use bounds on the singular values of the random matrix $S_C$.  As a shorthand, let $c_{\eps} = \frac{e^{\eps}+1}{e^{\eps}-1}$.  Then we have
\begin{align*}
    \left\| \frac{c\sqrt{d}}{n} S_{C}(\vec{y}_{C} - \vec{\underline{y}}_{C}) \right\|_{2}
    \leq{} & \frac{c\sqrt{d}}{n} \max_{C \subseteq [n]} \left\|S_{C}(\vec{y}_{C} - \vec{\underline{y}}_{C}) \right\|_{2} \\
    \leq{} & \frac{2c\sqrt{d}}{n} \max_{C \subseteq [n] \atop |C| = m}~\max_{\vec{y}_{C} \in \{-c_\eps,c_\eps\}^m} \left\| S_{C} \vec{y}_{C} \right\|_2 \\
    ={} &\frac{2c\sqrt{d}}{n} \max_{C \subseteq [n] \atop |C| = m}~\max_{\vec{y}_{C} \in \R^m \atop \|\vec{y}\|_{2} \leq c_{\eps} \sqrt{m}} \left\|  S_{C} \vec{y}_{C} \right\|_2 \\
    \leq{} &\frac{2c c_{\eps} \sqrt{m d}}{n} \max_{C \subseteq [n] \atop |C| = m}~\max_{\vec{y}_{C} \in \R^m \atop \|\vec{y}\|_{2} \leq 1} \left\|  S_{C} \vec{y}_{C} \right\|_2 \label{eq:} \stepcounter{equation} \tag{\theequation}
\end{align*}

For any $i \in C$, consider the random variable $\vec{s}_i\,' \sim N(0,I_{d\times d})$. The column vector $\vec{s}_i$ is identically distributed with $\frac{\vec{s}_i\,'}{\norm{\vec{s}_i\,'}_2}$. By standard concentration arguments, there is a constant $c'$ such that $\min_i \norm{\vec{s}_i\,'}_2^2 \geq d - c'\sqrt{d \ln m}$ with probability $\geq 299/300$. When $d > 4(c')^2\ln m$, we have $\min_i \norm{\vec{s}_i\,'}_2^2 > d/2$.\footnote{In the case where $d < 4(c')^2\ln m$, we bound the error from manipulation by $2cc_\eps m \sqrt{d}/n = O(m\sqrt{\log n}/\eps n)$.} Hence,
\begin{align*}
\eqref{eq:} \leq{} &\frac{2c c_{\eps} \sqrt{m d}}{n} \max_{C \subseteq [n] \atop |C| = m}~\max_{\vec{y}_{C} \in \R^m \atop \|\vec{y}\|_{2} \leq 1} \max_{i \in C} \frac{1}{\norm{\vec{s}_i\,'}_2} \left\|  S'_{C} \vec{y}_{C} \right\|_2 \\
\leq{} &\frac{c c_{\eps} \sqrt{8m }}{n} \max_{C \subseteq [n] \atop |C| = m}~\max_{\vec{y}_{C} \in \R^m \atop \|\vec{y}\|_{2} \leq 1} \left\|  S'_{C} \vec{y}_{C} \right\|_2\\
=& \frac{c c_{\eps} \sqrt{8m }}{n} \max_{C \subseteq [n] \atop |C| = m}~ \norm{S'_C}_2
\end{align*}
We apply Lemma \ref{lem:random-matrix-product} then choose $k=O(\ln n)$ to bound $\norm{S'_C}_2$ by $O(\sqrt{d\ln n} + \sqrt{m\ln n})$ with probability $\geq 299/300$. A union bound completes the proof.
\end{proof}

Observe that the manipulation error matches that of the lower bound in Theorem \ref{thm:binary-attack} for Bernoulli estimation up to a logarithmic factor.

\subsection{Uniformity Testing}
\label{sec:raptor}

In this problem, each user has data $x_i \in [d]$ sampled from a distribution $\bP$. If $\bP=\bU$, then a protocol for this problem should output ``uniform'' with probability $\geq 99/100$. If $\norm{\bP-\bU}_1 > \alpha$, then it should output ``not uniform'' with probability $\geq 99/100$. Smaller values of $\alpha$ are desirable.

We consider the $\raptor$ protocol, introduced by \cite{AcharyaCFT19}. It divides users into $G$ groups each of size $n/G$ (where $G$ is a parameter). In each group $g$,
\begin{enumerate}
    \item Sample public set $S \in \{ S \subset [d]~|~|S|=d/2\}$ uniformly at random.
    \item Each user assigns $x'_i \gets +1$ if $x_i \in S$ and otherwise $x'_i \gets -1$
    \item Each user $i$ reports $y_i \gets \rrr_\eps(x'_i)$ to the aggregator
    \item The aggregator computes the average of the messages: $\hat{\mu}_g \gets \frac{G}{n}\sum y_i$.
\end{enumerate}
If there is some $\hat{\mu}_g \gtrapprox \sqrt{\frac{1}{\eps^2 n}} + \frac{m}{\eps n} $, the aggregator reports ``not uniform.'' Otherwise, it reports ``uniform.''

\begin{thm}
\label{thm:raptor}
There is a choice of parameter $G$ such that, for any $\eps \in (0,1)$, any positive integers $m\leq n$, and any adversary $M$, the following holds with probability $\geq 99/100$
\[
\mathrm{Manip}_{m,n}(\raptor_{\eps}, \bU, M) = \mathrm{``uniform"}
\]
and, when $\norm{\bP-\bU}_1 \geq \alpha$ for some $\alpha = O \left(\sqrt{\frac{d}{\eps^2 n}} + \frac{m \sqrt{d}}{\eps n} \right)$, the following also holds with probability $\geq 99/100$
\[
\mathrm{Manip}_{m,n}(\raptor_{\eps}, \bP, M) = \mathrm{``not~uniform"}
\]
\end{thm}

\begin{proof}[Proof Sketch.]
Consider any $g \in [G]$. When $\norm{\bP-\bU}_1 \geq \sqrt{10d}\cdot \alpha$, a lemma by \cite{AcharyaCFT19} implies that, with at least some constant probability over the randomness of $S$, $\left| \pr{x \sim \bP}{x \in S} - \half \right| \gtrapprox \alpha$. For $\alpha \gtrapprox \sqrt{G/\eps^2 n} + mG/\eps n$, $\rr_\eps$ will provide an estimate of $\pr{x \sim \bP}{x \in S}$ that is larger than $\half + \alpha/2$. But when $\bP=\bU$, the protocol will give an estimate of $\pr{x \sim \bP}{x \in S}$ that is less than $\half + \alpha / 2$. This means there is a threshold test that has a constant probability of succeeding. The $G$ repetitions serve to increase the success probability to $99/100$. This completes the proof.
\end{proof}
Observe that the bound $\alpha = O(\frac{m \sqrt{d}}{\eps n} + \sqrt{\frac{d}{\eps^2 n}} )$ matches the lower bound of Theorem \ref{thm:uniformity-attack} up to logarithmic factors. We give the full details of the protocol in Appendix \ref{apdx:raptor}.

\subsection{Heavy Hitters}
\label{sec:hh-protocol}
In this problem, each user has data $x_i \in [d]$. The objective is to find a small subset $L$ of the universe that contains every element $j\in [d]$ such that $\mathit{freq}_j(\vec{x})>\alpha$. Because there are $1/\alpha$ heavy hitters, the size of $L$ should be $O(1/\alpha)$.

We consider the protocol $\hh$ described in \cite{BassilyNST17}.\footnote{In \cite{BassilyNST17}, the protocol is called \texttt{Bitstogram}.}

\begin{enumerate}
    \item Sample public hash function $h:[d]\to[k]$ uniformly from a universal family ($k\ll d$ is a protocol parameter). Also sample $\pi$ uniformly from partitions of $[n]$ into groups of size $n / \log_2 d$. Intuitively, users in group $g$ will communicate the $g$-th bit of their data value to the aggregator

    \item Each user $i$ in each group $g$:
    \begin{enumerate}
        \item obtains $\mathit{bit}(g,x_i)$, the $g$-th bit in the binary representation of $x_i$.
    
        \item computes $x'_i \gets 2\cdot h(x_i) - \mathit{bit}(g,x_i)$.

        \item reports $y_i \gets \hstr(x'_i)$ to the aggregator.
    \end{enumerate}
    
    \item The aggregator iterates through each $j' \in [k]$ and constructs $L_{j'}$ in the following manner:
    \begin{enumerate}
        \item Iterate through $g \in \log_2 d$. At each step $(j',g)$, gather the messages from group $g$ then use $\hsta$ to obtain an approximate histogram over $2k$. If the estimated frequency of $2\cdot j' - 1$ is larger than that of $2\cdot j'$, then set $z_{j',g} \gets 1$ and otherwise $z_{j',g} \gets 0$.
        \item $L_{j'} \gets$ the number represented in binary by $z_{j',1},\dots,z_{j',\log_2 d}$
    \end{enumerate}
    \item The aggregator reports $L \gets (L_1,\dots,L_{k})$ as heavy hitters
\end{enumerate}

The size of $L$ is at most $k$ and the time spent by the aggregator to construct $L$ is $O(nk^2\log d)$ (from $k\log_2 d$ executions of $\hsta$). An upper bound on error under manipulation follows from Theorem \ref{thm:hst-linfty}, taking care to adjust the number of bins to $2k$ and the number of users to $n/\log_2 d$.

\begin{thm}
\label{thm:hh-under-attack}
For any $\eps \in (0,1)$, any positive integers $m\leq n$, any $\vec{x}=(x_1,\dots,x_n) \in [d]^n$, and any adversary $M$, if we execute $L \gets \mathrm{Manip}_{m,n}(\hh_\eps,\vec{x},M)$ with parameter $k\gets 300n^2$, then with probability $\geq 99/100$, $L$ contains all $j$ such that $\mathit{freq}_j(\vec{x})>\alpha$ where
\[
\alpha = O\left(\sqrt{\frac{(\log d)\cdot \log (n\log d)}{\eps^2 n}} + \frac{m\log d}{\eps n}  \right)
\]
\end{thm}

\begin{proof}[Proof Sketch]
For any group $g$, let $\vec{x}^{(g)}$ denote the data of users in group $g$. We first argue that three undesirable events occur with low probability.
\begin{itemize}
    \item For some $g$, $|\mathit{freq}(\vec{x}) - \mathit{freq}(\vec{x}^{(g)})|_\infty \gtrapprox \sqrt{(\log d) / n}$. By Hoeffding's inequality and a union bound over all groups, this happens with probability $\leq 1/300$.
    \item Two data values that appear in $\vec{x}$ collide. Due to the size of $k$, this happens with probability $\leq 1/300$.
    \item For some $g$, the error of the private histogram is too large. Specifically, there is a value $\alpha_0\approx \sqrt{\frac{(\log d)\cdot \log (n\log d)}{\eps^2 n} } + \frac{m\log d}{\eps n} $ such that $\norm{\hat{\mu}^{(g)} - \mathit{freq}(\vec{x}^{(g)})}_\infty > \alpha_0$. From Theorem \ref{thm:hst-linfty} and a union bound over all groups, this event happens with probability $\leq 1/300$.
\end{itemize}
The remainder of the proof sketch assumes these events have not occurred.

We fix any $j \in [d]$ and any $g\in [\log_2 d]$. We will argue that if $j$ is a heavy hitter, then the aggregator will reconstruct the $g$-th bit of $j$. Let $\pi(g)$ be the ordered set of users in group $g$ and let $\vec{x}~'$ denote the vector $(x'_i)_{i \in \pi(g)}$.

Suppose $\mathit{freq}(j, \vec{x})=\tau$. Because there are no collisions between hashes, it must be the case that $\mathit{freq}(2h(j)-\mathit{bit}(g,j), \vec{x}~') \gtrapprox \tau-\sqrt{(\log d)/n}$ and $\mathit{freq}(2h(j)-1+\mathit{bit}(g,j), \vec{x}~')=0$. The aggregator estimates these frequencies up to simultaneous error $\alpha_0$. So, when $\tau > \alpha \approx \sqrt{(\log d) /n} + 2\alpha_0$, the estimate of $\mathit{freq}(2h(j)-\mathit{bit}(g,j), \vec{x}~')$ exceeds that of $\mathit{freq}(2h(j)-1+\mathit{bit}(g,j), \vec{x}~')$. This means the aggregator will assign $z_{h(j), g} \gets \mathit{bit}(g,j)$.
\end{proof}

We remark that the above sketch and analysis are for the simplest version of $\hh$, in which $k = O(n^2)$. In \cite{BassilyNST17}, the authors show that $k=O(1/\alpha) = \tilde{O}(\sqrt{n})$ suffices, achieving a smaller list and faster running time. We provide the full details of \hh~ in Appendix \ref{apdx:hh}.
    
    \section{Suboptimal Protocols} \label{sec:badprotocols}
In this section we demonstrate that there exist protocols with optimal error absent manipulation ($m = 0$) that perform quite poorly in the presence of manipulation $(m > 0$), thereby showing that a careful choice of protocols was necessary to achieve optimal robustness in Section~\ref{sec:protocols}.  

Intuitively, the protocols in Section~\ref{sec:protocols} achieve optimal robustness because they use public randomness to significantly constrain the choices of the corrupted users, and we argue that if we allow users to generate the randomness themselves, which has no effect on the protocol absent manipulation, then the protocol becomes much less robust.  

We can sketch an example of this phenomenon for frequency estimation, although essentially the same phenomenon arises in all of the problems we study.  Consider the following variant of the frequency estimation protocol:
\begin{enumerate}
    \item Each user chooses a uniformly random vector $\vec{s}_{i} \in \{\pm 1\}^d$.
    \item Each user samples $\gamma_i \gets \rrr_\eps(\vec{s}_{i,x_i})$ and reports the message $\vec{y}_i \gets \gamma_i \vec{s}_i \in \{ \pm \frac{e^\eps + 1}{e^\eps - 1} \}^d$.
    \item The aggregator outputs $\hat\mu \gets \frac{1}{n} \sum_{i=1}^{n} \vec{y}_i$.
\end{enumerate}
One can verify that when all users follow the protocol honestly, the distribution of the output $\hat\mu$ is identical to that of the protocol $\hst_{\eps}$.  Therefore, when users are honest, with high probability we have $\| \hat\mu - \mathit{freq}(\vec{x}) \|_1 = O(\sqrt{d^2 / \eps^2 n})$.

However, because the adversary can have the corrupted users report arbitrary vectors in $\{ \pm \frac{e^{\eps}+1}{e^{\eps}-1} \}^d$, and adversary who corrupts the first $m$ users can introduce error on the order of 
$$
\max_{\vec{y}_{1},\dots,\vec{y}_{m} \in \{\pm 1\}^d} \left\| \frac{1}{n} \sum_{i=1}^{m} \left( \frac{e^{\eps}+1}{e^{\eps}-1} \right) \vec{y}_i \right\|_{1} =  \left( \frac{e^{\eps}+1}{e^{\eps}-1} \right) \frac{md}{n} = \Omega\left( \frac{md}{\eps n} \right)
$$
In contrast, when we use the protocol $\hst_{\eps}$, we were able to show that the adversary could only introduce error $O(\frac{m \sqrt{d}}{\eps n})$.
    
    
    
    
    \subsection*{Acknowledgments}
    \addcontentsline{toc}{section}{Acknowledgments}

    Part of this work was done while the
    authors were visiting the Simons Institute for Theory of
    Computing.  AC and JU were supported by NSF grants CCF-1718088,
    CCF-1750640, and CNS-1816028.  JU was also supported by a Google
    Faculty Research Award.
    AS was supported by NSF award CCF-1763786
    and a Sloan Foundation Research Award.
    The authors are grateful to Gautam
    Kamath, Seth Neel, and Aaron Roth for helpful discussions.
    The authors are grateful for Jack Doerner for help preparing the figures.    
    
    \addcontentsline{toc}{section}{References}
    \bibliographystyle{alpha}
    \bibliography{refs.bib}
    
    \appendix

    \section{Proofs for Section \ref{sec:binary-attack}}
\label{apdx:binary-attack}
In this section, we prove the key technical claim from Section \ref{sec:binary-attack}. Let $p(m,n)=\half + \frac{m}{2n} + \sqrt{\frac{1}{2n} \ln 6}$

\begin{clm}[Claim \ref{clm:similar-binomials}, restated]
\label{clm:similar-binomials-apdx}
For all $n \geq 931$ and $m \leq n/8$, if $\mathbf{W} \sim m + \Bin(n-m, \frac{1}{2})$ and $\mathbf{W}^+ \sim \Bin\left(n,p(m,n) \right)$, then for any $\cW \subseteq [n]$,
\[
\pr{}{\mathbf{W}^+ \in \cW} \leq 51 \cdot \pr{}{\mathbf{W} \in \cW} + \frac{1}{3}
\]
\end{clm}

Claim \ref{clm:similar-binomials-apdx} is immediate from two intermediary claims presented below.

\begin{clm}
\label{clm:bounded-tail}
Fix $n \geq \frac{128}{49} \ln 6$, and $0\leq m \leq n/8$. If $\mathbf{W}^+ \sim \Bin(n, p(m,n))$, then
\[
\pr{}{\mathbf{W}^+ < \frac{n+m}{2} ~\mathrm{or}~ \mathbf{W}^+ > \frac{n+m}{2} + \sqrt{2n \ln 6} } < \frac{1}{3}
\]
\end{clm}
\begin{proof}
The expected value of $\mathbf{W}^+$ is $\frac{n+m}{2} + \sqrt{\frac{n\ln 6}{2}}$. From Hoeffding's Inequality,
\[
\pr{}{\left|\mathbf{W}^+ - \left(\frac{n+m}{2} + \sqrt{\frac{n\ln 6}{2}} \right) \right| > \sqrt{n \cdot \frac{\ln 6}{2}}} < \frac{1}{3}
\]
This concludes the proof.
\end{proof}

\begin{clm}
\label{clm:bounded-ratio}
If $\mathbf{W} \sim m+\Bin(n-m, \half)$ and $\mathbf{W}^+ \sim \Bin(n, p(m,n))$. For all $n \geq 931$, $m \leq n/8$, and $0 < k \leq \sqrt{2n \ln 6}$,
\[
\frac{\pr{}{\mathbf{W}^+ = \frac{n+m}{2} + k}}{\pr{}{\mathbf{W} = \frac{n+m}{2} + k}} \leq 51
\]
\end{clm}

\begin{proof}
We prove the claim by a direct calculation using the probability mass function of the binomial distribution. As shorthand, we use $p=p(m,n)$.
\begin{align*}
\frac{\pr{}{\Bin(n,p) = \frac{n+m}{2} + k}}{\pr{}{m + \Bin( n-m, \frac{1}{2}) = \frac{n+m}{2} +k} } &= \frac{\pr{}{\Bin(n,p) = \frac{n+m}{2} + k}}{\pr{}{\Bin( n-m, \frac{1}{2}) = \frac{n-m}{2} +k} }\\
&= \frac{ p^{\frac{n+m}{2}+k} (1-p)^{\frac{n-m}{2}-k} \cdot \frac{n!}{\left(\frac{n+m}{2}+k \right)! \left(\frac{n-m}{2}-k \right)! }}{2^{-(n-m)}\cdot \frac{(n-m)!}{\left(\frac{n-m}{2} + k\right)! \left(\frac{n-m}{2} - k\right)!}}\\
&= \tau \cdot \frac{n! \left(\frac{n-m}{2}+k \right)! }{(n-m)! \left(\frac{n+m}{2}+k \right)! } \label{eq:before-stirling}\stepcounter{equation} \tag{\theequation}
\end{align*}
where we introduce $\tau := 2^{n-m}\cdot  p^{\frac{n+m}{2}+k} \cdot (1-p)^{\frac{n-m}{2}-k}$. Recall Stirling's approximation $\sqrt{2\pi z} \left(\frac{z}{e}\right)^z \leq z! \leq e\sqrt{z} \left(\frac{z}{e}\right)^z$ for any integer $z$. We will use this in order to upper bound \eqref{eq:before-stirling}.

\begin{align*}
&\eqref{eq:before-stirling}\\
&\leq \frac{e^2}{2\pi}\cdot \tau \cdot \frac{n^n \left(\frac{n-m}{2}+k \right)^{\left(\frac{n-m}{2}+k \right)} }{(n-m)^{(n-m)} \left(\frac{n+m}{2}+k \right)^{\left(\frac{n+m}{2}+k \right)} } \cdot \sqrt{\frac{n}{n-m} \cdot \frac{\frac{n-m}{2} + k}{\frac{n+m}{2} + k} } \\
&= \frac{e^2}{2\pi}\cdot \tau \cdot \frac{n^{\frac{n-m}{2}-k} \left(\frac{n-m}{2}+k \right)^{\left(\frac{n-m}{2}+k \right)} }{(n-m)^{(n-m)} \left(\frac{n+m+2k}{2n} \right)^{\left(\frac{n+m}{2} +k \right)} } \cdot \sqrt{\frac{n}{n-m} \cdot \frac{\frac{n-m}{2} + k}{\frac{n+m}{2} + k} } \\
&= \frac{e^2}{2\pi}\cdot \tau \cdot \left( \frac{n}{n-m} \right)^{\frac{n-m}{2}-k} \cdot \frac{ \left(\frac{1}{2}+\frac{k}{n-m} \right)^{\left(\frac{n-m}{2}+k \right)} }{ \left(\frac{n+m+2k}{2n} \right)^{\left(\frac{n+m}{2} +k \right)} } \cdot \sqrt{\frac{n}{n-m} \cdot \frac{\frac{n-m}{2} + k}{\frac{n+m}{2} + k} } \\
&\leq \sqrt{\frac{2e^4}{7\pi^2}} \cdot \tau \cdot \left( \frac{n}{n-m} \right)^{\frac{n-m}{2}-k} \cdot \frac{ \left(\frac{1}{2}+\frac{k}{n-m} \right)^{\left(\frac{n-m}{2}+k \right)} }{ \left(\frac{n+m+2k}{2n} \right)^{\left(\frac{n+m}{2} +k \right)} } \tag{$0 \leq m \leq n/8$}\\
&= \sqrt{\frac{2e^4}{7\pi^2}} \cdot 2^{n-m} \cdot \left( \frac{n(1-p)}{n-m} \right)^{\frac{n-m}{2}-k} \cdot \frac{ \left(\frac{1}{2}+\frac{k}{n-m} \right)^{\left(\frac{n-m}{2}+k \right)} }{ \left(\frac{n+m+2k}{2np} \right)^{\left(\frac{n+m}{2} +k \right)} } \\
&= \sqrt{\frac{2e^4}{7\pi^2}} \cdot \left( \frac{2n(1-p)}{n-m} \right)^{\frac{n-m}{2}-k} \cdot \frac{ \left(1+\frac{2k}{n-m} \right)^{\left(\frac{n-m}{2}+k \right)} }{ \left(\frac{n+m+2k}{2np} \right)^{\left(\frac{n+m}{2} +k \right)} } \\
&= \sqrt{\frac{2e^4}{7\pi^2}} \cdot \left( \frac{2n(1-p)}{n-m} \right)^{\frac{n-m}{2}-k} \cdot \left(1+\frac{2k}{n-m} \right)^{\left(\frac{n-m}{2}+k \right)} \cdot \left(\frac{2np}{n+m+2k} \right)^{\left(\frac{n+m}{2} +k \right)} \stepcounter{equation} \tag{\theequation} \label{eq:bc-shorthands}
\end{align*}
As shorthands, we use $b = \sqrt{\frac{2e^4}{7\pi^2}}$ and $c = \sqrt{\frac{\ln 6}{2}}$:
\begin{align*}
&\eqref{eq:bc-shorthands}\\
&= b \cdot \left( \frac{n-m-2c\sqrt{n}}{n-m} \right)^{\frac{n-m}{2}-k} \cdot \left(1+\frac{2k}{n-m} \right)^{\left(\frac{n-m}{2}+k \right)} \cdot \left(\frac{n+m+2c\sqrt{n}}{n+m+2k} \right)^{\left(\frac{n+m}{2} +k \right)} \\
&= b\cdot \underbrace{\left(1+ \frac{2c\sqrt{n}}{n-m-2c\sqrt{n}} \right)^{- \left(\frac{n-m}{2}-k\right)}}_{(i)} \cdot \underbrace{\left(1+\frac{2k}{n-m} \right)^{\left(\frac{n-m}{2}+k \right)}}_{(ii)} \cdot \underbrace{\left(1 + \frac{2c\sqrt{n}-2k}{n+m+2k} \right)^{\left(\frac{n+m}{2} +k \right)}}_{(iii)} \label{eq:ln-approx} \stepcounter{equation} \tag{\theequation}
\end{align*}

For any $z \in (0,1)$, recall that $z - \frac{z^2}{2} \leq \ln (1+z) \leq z$ for $z \in (0,1)$. Because our values of $n,m,c$ ensure that $\frac{2c\sqrt{n}}{n-m-2c\sqrt{n}}$ and $\frac{2k}{n-m}$ both lie in the interval $(0,1)$, we upper bound terms $(i),(ii)$ in \eqref{eq:ln-approx}.

\begin{align*}
(\ref{eq:ln-approx},i) &= \left(1 + \frac{2c\sqrt{n}}{n-m-2c\sqrt{n}} \right)^{-\left( \frac{n-m}{2}-k \right)}\\
&= \exp \left(-\left(\frac{n-m}{2}-k\right) \ln \left(1+ \frac{2c\sqrt{n}}{n-m-2c\sqrt{n}} \right) \right)\\
&\leq \exp \left(-\left(\frac{n-m}{2}-k\right) \cdot \left( \frac{2c\sqrt{n}}{n-m-2c\sqrt{n}} - \frac{2c^2 n}{(n-m-2c\sqrt{n})^2} \right) \right) \\
&= \exp \left( -c\sqrt{n} \cdot \frac{n-m-2k}{n-m-2c\sqrt{n}} + c^2n \cdot \frac{n-m-2k}{(n-m-2c\sqrt{n})^2} \right)\\
&\leq \exp \left( -c\sqrt{n} \cdot \frac{n-m-2k}{n-m-2c\sqrt{n}} + c^2 \cdot \frac{n-m}{n-m-2c\sqrt{n}} \right) \label{eq:ln-bound-1} \stepcounter{equation} \tag{\theequation} 
\end{align*}

\begin{align*}
(\ref{eq:ln-approx},ii) &= \left(1+\frac{2k}{n-m}\right)^{\left(\frac{n-m}{2}+k \right)}\\
&= \exp \left(\left(\frac{n-m}{2}+k \right) \ln \left(1+\frac{2k}{n-m}\right) \right)\\
&\leq \exp \left(\left(\frac{n-m}{2}+k \right) \cdot \frac{2k}{n-m}  \right)\\
&= \exp \left(k + \frac{2k^2}{n-m} \right)\\
&\leq \exp \left(k + \frac{2c^2n}{n-m} \right) \stepcounter{equation}
\tag{\theequation} \label{eq:ln-bound-2}
\end{align*}

In the case where $k \leq c\sqrt{n}$, the ratio $\frac{2c\sqrt{n} - 2k}{n+m+2k}$ is in $(0,1)$. Hence,
\begin{align*}
(\ref{eq:ln-approx},iii) &= \left(1+\frac{2c\sqrt{n} - 2k}{n+m+2k} \right)^{\left(\frac{n+m}{2} +k \right)}\\
&= \exp \left(\left(\frac{n+m}{2} +k \right) \ln \left(1 + \frac{2c\sqrt{n} - 2k}{n+m+2k} \right) \right)\\
&\leq \exp \left(\left(\frac{n+m}{2} +k \right) \cdot \frac{2c\sqrt{n} - 2k}{n+m+2k} \right) \\
&= \exp(-k+c\sqrt{n}) \stepcounter{equation} \tag{\theequation} \label{eq:ln-bound-3} 
\end{align*}

When $c\sqrt{n} < k \leq 2c \sqrt{n}$, we note that (\ref{eq:ln-approx},iii) is equivalent to $\left(1+\frac{2k- 2c\sqrt{n}}{n+m+2c\sqrt{n}} \right)^{-\left(\frac{n+m}{2} +k \right)}$. The ratio $\frac{2k - 2c\sqrt{n}}{n+m+2c\sqrt{n}}$ is in $(0,1)$
\begin{align*}
(\ref{eq:ln-approx},iii) &= \left(1+\frac{2k- 2c\sqrt{n}}{n+m+2c\sqrt{n}} \right)^{-\left(\frac{n+m}{2} +k \right)}\\
&= \exp \left(-\left(\frac{n+m}{2} +k \right) \ln \left(1+\frac{2k- 2c\sqrt{n}}{n+m+2c\sqrt{n}} \right) \right)\\
&\leq \exp \left(-\left(\frac{n+m}{2} +k \right) \left(\frac{2k- 2c\sqrt{n}}{n+m+2c\sqrt{n}} - \frac{2(k- c\sqrt{n})^2}{(n+m+2c\sqrt{n})^2} \right) \right) \\
&= \exp \left(- \frac{n+m+2k}{n+m+2c\sqrt{n}}\cdot (k- c\sqrt{n}) + \frac{n+m+2k}{(n+m+2c\sqrt{n})^2} \cdot (k- c\sqrt{n})^2  \right) \\
&\leq \exp \left(- \frac{n+m+2k}{n+m+2c\sqrt{n}}\cdot (k- c\sqrt{n}) + \frac{n+m+2k}{(n+m+2c\sqrt{n})^2} \cdot nc^2 \right) \tag{$k \leq 2c\sqrt{n}$}\\
&< \exp \left(- k + \frac{n+m+2k}{n+m+2c\sqrt{n}}\cdot c\sqrt{n} + \frac{n+m+2k}{(n+m+2c\sqrt{n})^2} \cdot nc^2 \right) \tag{$k > c\sqrt{n}$}\\
&\leq \exp \left(- k + \frac{n+m+2k}{n+m+2c\sqrt{n}}\cdot c\sqrt{n} + \frac{n+m+2k}{n+m+2c\sqrt{n}} \cdot \frac{1}{4} \right)\\
&\leq \exp \left(- k + \frac{n+m+2k}{n+m+2c\sqrt{n}}\cdot c\sqrt{n} + \half \right) \stepcounter{equation} \tag{\theequation} \label{eq:ln-bound-3'}
\end{align*}

The final inequality comes from the upper bound on $k$. Because \eqref{eq:ln-bound-3'} dominates \eqref{eq:ln-bound-3}, we will use it to form our upper bound. Taking \eqref{eq:ln-bound-1}, \eqref{eq:ln-bound-2}, \eqref{eq:ln-bound-3'} together, we have

\begin{align*}
&\eqref{eq:ln-approx}\\
&\leq be^{1/2}\cdot \exp \left( c^2 \cdot \left( \frac{2n}{n-m} + \frac{n - m}{n-m-2c\sqrt{n}} \right) + c\sqrt{n} \cdot \left( \frac{n+m+2k}{n+m+2c\sqrt{n}} - \frac{n-m-2k}{n-m-2c\sqrt{n}} \right) \right)\\
&= be^{1/2} \cdot \exp \left( c^2 \cdot \left( \frac{2n}{n-m} + \frac{n - m}{n-m-2c\sqrt{n}} \right) + c\sqrt{n} \cdot \frac{4n(k-c\sqrt{n})}{n^2-m^2 -4c \sqrt{n}(m+c\sqrt{n})} \right)\\
&\leq be^{1/2}\cdot \exp \left( c^2 \cdot \left( \frac{2n}{n-m} + \frac{n - m}{n-m-2c\sqrt{n}} + \frac{4n^{3/2}}{n^2-m^2 -4c \sqrt{n}(m+c\sqrt{n})} \right) \right) \tag{$k \leq 2c\sqrt{n}$}\\
&\leq be^{1/2}\cdot \exp \left( c^2 \cdot \left( \frac{16}{7} + \frac{7n}{7n-16c\sqrt{n}} + \frac{4n^{3/2}}{\frac{15}{16}n^2 -\half c n^{3/2} - 4c^2n} \right) \right) \tag{$m \leq n/8$}\\
&\leq be^{1/2}\cdot \exp \left( c^2 \cdot \left( \frac{16}{7} + \frac{8}{7} + \frac{4n^{3/2}}{\frac{15}{16}n^2 -\half c n^{3/2} - 4c^2n} \right) \right) \tag{$n \geq \left(\frac{96}{7}\right)^2 c$}\\
&\leq be^{1/2}\cdot \exp \left( c^2 \cdot \left(\frac{24}{7} + \frac{1}{7} \right) \right) \stepcounter{equation} \tag{\theequation} \label{eq:apply-n}\\
&< 51 \tag{Defn. of $b,c$}
\end{align*}

To arrive at \eqref{eq:apply-n}, we observe that all $n \geq 931$ are solutions to the quadratic inequality $\frac{4n^{3/2}}{\frac{15}{16}n^2 -\half c n^{3/2} - 4c^2n} \leq 1/7$. This concludes the proof.
\end{proof}

    \section{Proofs for Section \ref{sec:finite-attack}}
\label{apdx:finite-attack}

For any integer $d>2$ and algorithm $R:[d]\to \cY$, let $R(\bU)$ denote the distribution over $\cY$ induced by sampling $\hat{x}$ from the uniform distribution over $[d]$ and then sampling a message from $R(\hat{x})$. For any set $H\subset [d]$, let $R(\bU_H)$ denote the distribution over $\cY$ induced by sampling $\hat{x}$ from the uniform distribution over $H$ and then executing $R(\hat{x})$. In this notation, $Q_{H,R}$ is the algorithm which samples from $R(\bU_H)$ when given $+1$, but samples from $R(\bU_{\overline{H}})$ when given $-1$.

In this section, we provide two bounds on the privacy parameter $\eps'$ of $Q_{H,R}$ when $H$ is uniformly chosen. The first bound is $O (\eps \sqrt{(\log (|\cY|\cdot |\vec{R}|_{\neq})) / d} )$ (Lemma \ref{lem:dependent-amplification}), where $|\cY|$ is the size of the message universe and $|\Vec{R}|_{\neq}$ is the number of unique randomizers. The second is $O\big(\eps \sqrt{(\log n) / d}\big)$ (Lemma \ref{lem:independent-amplification}). We note that the second bound has no dependence on the specification of the randomizers, which may make it looser than the first bound.

The key to the analysis is to argue that, for most messages $y$ and a uniformly random $H$, the log-odds ratio $\ln (\pr{}{R(\bU_H)=y} / \pr{}{R(\bU)=y})$ is roughly $\eps / \sqrt{d}$. To this end, we introduce the following definition:

\begin{defn}[Leaky Messages]
\label{defn:leaky-message}
For any $H \subset [d]$ with size $d/2$ and any local randomizer $R:[d] \rightarrow \cY$, a message $y\in \cY$ is $v$-\emph{leaky with respect to} $H,R$ when
\begin{equation}
\label{eq:leaky-message}
\left|\ln \frac{\pr{}{R(\bU_H) = y}}{\pr{}{R(\bU) = y}} \right| > v
\end{equation}
\end{defn}

Next we show that when $y$ is some fixed message and $H$ is uniformly random, $y$ is $\approx (\eps / \sqrt{d})$-leaky with respect to $H,R$ with low probability.
\begin{clm}
\label{clm:leaky-message-unlikely}
Fix any $\eps > 0$, any $\beta \in (0,1)$, any $d > 4(e^{2\eps}-1)^2 \ln \frac{2}{\beta}$, any $\eps$-private $R:[d] \rightarrow \cY$. For any message $y \in \cY$, if $H$ is chosen uniformly from subsets of $[d]$ with size $d/2$, then
\[
\pr{}{y \textrm{ is not } (e^{2\eps}-1) \sqrt{\frac{4}{d} \ln \frac{2}{\beta}} \textrm{-leaky w.r.t. } H,R } \geq 1- \beta
\]
\end{clm}

\begin{proof}
By the definition of leaky message, we must show that the following must hold with probability $\geq 1-\beta$ over the randomness of $H$.
\begin{equation}
\label{eq:ratio-objective}
\exp\left(-(e^{2\eps}-1) \sqrt{\frac{4}{d} \ln \frac{2}{\beta}} \right) \leq \frac{\pr{}{R(\bU_H) = y}}{\pr{}{R(\bU) = y}} \leq \exp\left((e^{2\eps}-1) \sqrt{\frac{4}{d} \ln \frac{2}{\beta}} \right)
\end{equation}

Observe that 
\begin{align*}
\pr{}{R(\bU) = y} &= \sum_{j=1}^d \pr{}{R(j) = y} \cdot \pr{x \sim \bU}{x = j}\\
    &= \sum_{j=1}^d \pr{}{R(j) = y} \cdot \frac{1}{d} \tag{Defn. of $\bU$}
\end{align*}
Also observe that, for any fixed choice of $H$,
\begin{equation*}
\pr{}{R(\bU_H) = y} = \sum_{i=1}^{d/2} \pr{}{R(h_i) = y} \cdot \frac{2}{d} \tag{By construction}
\end{equation*}
Now we may write
\begin{equation}
\label{eq:probability-ratio}
\frac{\pr{}{R(\bU_H) = y}}{\pr{}{R(\bU) = y}} = \frac{\frac{2}{d}\sum_{i=1}^{d/2} \pr{}{R(h_i) = y}}{\frac{1}{d}\sum_{j=1}^d \pr{}{R(j) = y}}
\end{equation}

For a uniformly random $H$, observe that each term in the numerator of \eqref{eq:probability-ratio} is a random variable that lies in the interval $(e^{-\eps}\max_j \pr{}{R(j)=y}, e^{\eps}\max_j \pr{}{R(j)=y})$, due to the $\eps$-privacy guarantee of $R$. We use the following version of Hoeffding's inequality for samples without replacement.
\begin{lem}[\cite{Hoeffding63}]
Given a set $\vec{p}=\{p_1,\dots,p_N\} \in \R^N$ such that $p_i \in (c,c')$, if the subset $\vec{x}=\{x_1,\dots,x_n\}$ is constructed by uniformly sampling without replacement from $\vec{p}$, then
\[
\pr{}{\frac{1}{n} \sum_{i=1}^n x_i \leq \frac{1}{N} \sum_{i=1}^N p_i + (c'-c)\cdot \sqrt{ \frac{1}{2n} \log \frac{1}{\beta}}} \geq 1- \beta
\]
\end{lem}
Hence, the following is true with probability $1-\beta/2$:
\begin{align*}
\eqref{eq:probability-ratio} &\leq \frac{\frac{1}{d}\sum_{j=1}^d \pr{}{R(j) = y} + (e^\eps \max_j \pr{}{R(j)=y}- e^{-\eps}\max_j \pr{}{R(j)=y}) \sqrt{ \frac{1}{d} \ln \frac{2}{\beta}}}{\frac{1}{d}\sum_{j=1}^d \pr{}{R(j) = y}}\\
    &= 1 + (e^\eps-e^{-\eps}) \sqrt{\ln \frac{2}{\beta}} \cdot \frac{ \sqrt{ d } \cdot \max_j \pr{}{R(j)=y}}{\sum_{j=1}^d \pr{}{R(j) = y}}\\
    &\leq 1 + (e^\eps-e^{-\eps}) \sqrt{\ln \frac{2}{\beta}} \cdot \frac{ \sqrt{ d  }\cdot \max_j\pr{}{R(j)=y}}{d \cdot \min_j \pr{}{R(j) = y}}\\
    &= 1 + (e^\eps-e^{-\eps}) \sqrt{\frac{1}{d} \ln \frac{2}{\beta}} \cdot \frac{  \max_j\pr{}{R(j)=y}}{\min_j \pr{}{R(j) = y}}\\
    &\leq 1 + (e^{2\eps}-1) \sqrt{\frac{1}{d} \ln \frac{2}{\beta}} \tag{$R$ is $\eps$-private}\\
    &\leq \exp\left( (e^{2\eps}-1) \sqrt{\frac{1}{d} \ln \frac{2}{\beta}} \right) \label{eq:ratio-upper-bound} \stepcounter{equation} \tag{\theequation}
\end{align*}

By a completely symmetric argument, the following holds with probability $1-\beta/2$:
\begin{align*}
\eqref{eq:probability-ratio} &\geq 1 - (e^{2\eps}-1) \sqrt{\frac{1}{d} \ln \frac{2}{\beta}} \\
    &\geq \exp\left(-(e^{2\eps}-1) \sqrt{\frac{4}{d} \ln \frac{2}{\beta}} \right) \label{eq:ratio-lower-bound} \stepcounter{equation} \tag{\theequation}
\end{align*}
\eqref{eq:ratio-lower-bound} follows from the condition that $d > 4(e^{2\eps}-1)^2 \ln \frac{2}{\beta}$. \eqref{eq:ratio-objective} follows from \eqref{eq:ratio-upper-bound} and \eqref{eq:ratio-lower-bound} (through a union bound). This concludes the proof.
\end{proof}

Now we apply Claim \ref{clm:leaky-message-unlikely} to analyze the privacy of any $Q_{H,R_i}$.

\subsection{A protocol-dependent bound on $\eps'$}
Our first bound on the privacy parameters will be dependent on the structure of the initial randomizers $R_1,\dots,R_n$ from which the new randomizers $Q_{H,R_1},\dots,Q_{H,R_n}$ are derived. We use $|\cY|$ to denote the size of the message universe and $|\Vec{R}|_{\neq}$ to denote the number of unique randomizers.

The following is immediate from Claim \ref{clm:leaky-message-unlikely} by applying a union bound over all the unique randomizers in $\vec{R}$ and the message universe $\cY$:
\begin{coro}
\label{coro:dependent-amplification}
Fix any vector of $\eps$-private randomizers $\vec{R} = (R_1, \dots, R_n)$ (where every randomizer has the form $R_i : [d] \rightarrow \cY$) and any $d > 4(e^{2\eps}-1)^2 \ln (12 |\cY|\cdot |\vec{R}|_{\neq})$. Sample $H$ uniformly at random over subsets of $[d]$ with size $d/2$. The following is true with probability $\geq 5/6$ over the randomness of $H$:
\[
\forall y \in \cY~ \forall i \in [n]~ y \textrm{ is not } \left( (e^{2\eps}-1) \sqrt{\frac{4}{d} \ln 12|\cY|\cdot |\vec{R}|_{\neq} } \right) \textrm{-leaky w.r.t.}~ H, R_i
\]
\end{coro}

\begin{lem}
\label{lem:dependent-amplification}
Fix any $\eps$-locally private protocol $\Pi = (\vec{R}, A)$ (where every randomizer has the form $R_i: [d] \to \cY$) and $d > 4(e^{2\eps}-1)^2 \ln (12 |\cY|\cdot |\vec{R}|_{\neq})$. Sample $H$ uniformly at random over subsets of $[d]$ with size $d/2$. The following is true with probability $\geq 2/3$ over the randomness of $H$: all randomizers $\{Q_{H, R_i}\}_{i \in [n]}$ specified by Algorithm \ref{alg:randomizer-reduction} satisfy $\eps'$-privacy, where $$\eps' = (e^{2\eps}-1) \sqrt{\frac{16}{d} \ln \left(12|\cY| \cdot |\vec{R}|_{\neq} \right)}$$
\end{lem}

\begin{proof}
From Corollary \ref{coro:dependent-amplification}, all possible outputs of all randomizers are not leaky with probability $\geq 5/6$. More formally, for every $y \in \cY$ and $i\in[n]$,
\[
\left|\ln \frac{\pr{}{R_i(\bU_H) = y}}{\pr{}{R_i(\bU) = y}} \right| < \eps'/2
\]
By identical reasoning, with probability $\geq 5/6$,
\[
\left|\ln \frac{\pr{}{R_i(\bU_{\overline{H}}) = y}}{\pr{}{R_i(\bU) = y}} \right| < \eps'/2
\]

Recall the definition of $Q_{H,R_i}$: on input $+1$, it samples from $R_i(\bU_H)$ and, on input $-1$, it samples from $R_i(\bU_{\overline{H}})$. From a union bound, we can conclude that the log-odds ratio is at most $\eps'$ with probability $\geq 2/3$. This concludes the proof.
\end{proof}

\subsection{A protocol-independent bound on $\eps'$}
In this subsection, we obtain a bound on the amplified privacy that depends on the number of users in the protocol but not on the specification of the randomizers $\vec{R}$. If $H$ is drawn uniformly and $d$ is sufficiently large, then for most users, the probability that $R_i(\bU)$ is a leaky message is small. Let $\mathit{Leak}(v,H,R) = \{y \in \cY ~|~ y\textrm{ is }v\textrm{-leaky with respect to }H,R\}$


\begin{clm}
Fix any $\eps > 0$, any $\beta \in (0,1)$, any $d > 4(e^{2\eps}-1)^2 \ln \frac{2}{\beta}$, and any $n$ algorithms $R_1\dots R_n$ that are $\eps$-private. If $H$ is sampled uniformly from subsets of $[d]$ with size $d/2$, then the following holds with probability $\geq 5/6$ over the randomness of $H$:
\begin{equation}
\label{eq:good-h}
\forall i \in [n]~ \pr{}{R_i(\bU) \in \mathit{Leak}\left( (e^{2\eps}-1)\sqrt{\frac{4}{d}\ln \frac{2}{\beta}}, H, R_i \right)} < 6\beta n
\end{equation}
\end{clm}

\begin{proof}
To prove the claim, we show that for every $i\in[n]$, with probability at least $1-1/6n$ over the randomness of $H$,
\begin{equation}
\label{eq:leaky-h}
\pr{}{R_i(\bU) \in \mathit{Leak}\left( (e^{2\eps}-1)\sqrt{\frac{4}{d}\ln \frac{2}{\beta}}, H, R_i \right)} < 6\beta n
\end{equation}

Below, we use $\binom{[d]}{d/2}$ as shorthand for the subsets of $[d]$ with size $d/2$. We bound the expectation of the random variable:
\begin{align*}
\ex{H}{\pr{}{R_i(\bU) \in \mathit{Leak}(\dots, H, R_i)}} &= \sum_{H \in \binom{[d]}{d/2}} \binom{d}{d/2}^{-1} \cdot \pr{}{R_i(\bU) \in \mathit{Leak}(\dots, H, R_i)} \\
    &= \sum_{H \in \binom{[d]}{d/2}} \binom{d}{d/2}^{-1} \cdot \sum_{y \in \cY} \indic{y \in \mathit{Leak}(\dots,H,R_i)} \cdot \pr{}{R_i(\bU)=y} \\
    &= \sum_{y \in \cY} \sum_{H \in \binom{[d]}{d/2}} \binom{d}{d/2}^{-1} \cdot \indic{y \in \mathit{Leak}(\dots,H,R_i)} \cdot \pr{}{R_i(\bU)=y}\\
    &\leq \sum_{y \in \cY} \beta \cdot \pr{}{R_i(\bU)=y} \tag{Claim \ref{clm:leaky-message-unlikely}}\\
    &= \beta
\end{align*}
Markov's inequality implies that \eqref{eq:leaky-h} holds with probability $\geq 1-1/6n$.
\end{proof}

\eqref{eq:good-h} is a bound on the probability that $R_i(\bU)$ is leaky. Because $\vec{R}$ satisfies differential privacy, \eqref{eq:good-h} implies a bound on the probability that $R_i(\bU_H)$ is leaky.

\begin{coro}
\label{coro:independent-amplification}
Fix any $\eps > 0$, any $\beta \in (0,1)$, any $d > 4(e^{2\eps}-1)^2 \ln \frac{2}{\beta}$, and any $n$ algorithms $R_1\dots R_n$ that are $\eps$-private. If $H$ is sampled uniformly from subsets of $[d]$ with size $d/2$, then the following holds with probability $\geq 5/6$ over the randomness of $H$:
\begin{align*}
\forall i \in [n]~~ \pr{}{R_i(\bU) \in \mathit{Leak}\left( (e^{2\eps}-1)\sqrt{\frac{4}{d}\ln \frac{2}{\beta}}, H, R_i \right)} &< 6\beta n \\
\forall i \in [n]~~ \pr{}{R_i(\bU_H) \in \mathit{Leak}\left( (e^{2\eps}-1)\sqrt{\frac{4}{d}\ln \frac{2}{\beta}}, H, R_i \right)} &< 6e^\eps \beta n
\end{align*}
\end{coro}

The algorithm $Q_{H,R_i}$ reports either a sample from $R_i(\bU_H)$ or from $R_i(\bU_{\overline{H}})$. Having bounded the probability that either sample is leaky, we can now argue that $Q_{H,R_i}$ satisfies approximate differential privacy.

\begin{lem}
\label{lem:independent-amplification}
Fix any $\eps>0$, any $\delta,\beta \in (0,1)$, any $d > 4(e^{2\eps}-1)^2 \ln (12e^\eps n / \delta )$, and any $n$ algorithms that are $\eps$-private. If $H$ is sampled uniformly from subsets of $[d]$ with size $d/2$, then the following holds  with probability $> 2/3$ over the randomness of $H$: all randomizers $\{Q_{H, R_i}\}_{i \in [n]}$ specified by Algorithm \ref{alg:randomizer-reduction} satisfy $(\eps', \delta)$-privacy, where $\eps' = (e^{2\eps}-1) \sqrt{\frac{16}{d} \ln (24e^\eps n / \delta )}$.
\end{lem}

\begin{proof}
Define $\beta = \delta / (12e^\eps n)$ so that $\eps' = (e^{2\eps}-1) \sqrt{\frac{16}{d} \ln (2/\beta)}$. For every $Y \subseteq \cY$, the following holds with probability $>5/6$ by Corollary \ref{coro:independent-amplification}.
\begin{align*}
\pr{}{R_i(\bU_H) \in Y} &= \pr{}{R_i(\bU_H) \in Y - \mathit{Leak}(\eps'/2, H, R_i)} + \pr{}{R_i(\bU_H) \in Y \cap \mathit{Leak}(\eps'/2, H, R_i)}\\
    &\leq \pr{}{R_i(\bU_H) \in Y - \mathit{Leak}(\eps'/2, H, R_i)} + 6\beta e^\eps n \tag{Corollary \ref{coro:independent-amplification}}\\
    &= \pr{}{R_i(\bU_H) \in Y - \mathit{Leak}(\eps'/2, H, R_i)} + \delta/2 \tag{Value of $\beta$}\\
    &= \sum_{y \in Y- \mathit{Leak}(\eps'/2)} \pr{}{R_i(\bU_H) = y}  + \delta/2\\
    &\leq \sum_{y \in Y- \mathit{Leak}(\eps'/2)} \exp(\eps'/2) \cdot \pr{}{R_i(\bU) = y}  + \delta/2 \tag{Defn. \ref{defn:leaky-message}}\\
    &\leq \exp(\eps'/2)\cdot \pr{}{R_i(\bU) \in Y} + \delta/2\\
\pr{}{R_i(\bU) \in Y} &\leq \exp(\eps'/2)\cdot \pr{}{R_i(\bU_H) \in Y} + \delta/2 \tag{Symmetric steps}
\end{align*}
We take identical steps to show that the following holds with probability $> 5/6$ as well:
\begin{align*}
\pr{}{R_i(\bU_{\overline{H}}) \in Y} &\leq \exp(\eps'/2)\cdot \pr{}{R_i(\bU) \in Y} + \delta/2 \\
\pr{}{R_i(\bU) \in Y} &\leq \exp(\eps'/2)\cdot \pr{}{R_i(\bU_{\overline{H}}) \in Y} + \delta/2
\end{align*}

From basic composition and a union bound, the following holds with probability $>2/3$:
\begin{align*}
\pr{}{R_i(\bU_{\overline{H}}) \in Y} &\leq \exp(\eps')\cdot \pr{}{R_i(\bU_H) \in Y} + \delta \\
\pr{}{R_i(\bU_H) \in Y} &\leq \exp(\eps')\cdot \pr{}{R_i(\bU_{\overline{H}}) \in Y} + \delta
\end{align*}
Recall that $Q_{H,R_i}$ samples from $R_i(\bU_H)$ on input $+1$ and from $R_i(\bU_{\overline{H}})$ on input $-1$. Hence, $Q_{H,R_i}$ satisfies $\eps',\delta$ privacy. This concludes the proof.
\end{proof}


    
    \section{Construction and Analysis of Protocols from Section \ref{sec:protocols}}

\subsection{Construction and Analysis of $\esti$}
\label{apdx:esti}

The protocol $\esti_{n,d,\eps}$ consists of the $n$ randomizers $(\estir_{n,d, \eps, i})_{i \in [n]}$ and the aggregator $\estia_{n,d,\eps}$; see Algorithms \ref{alg:rir} and \ref{alg:ria} for the pseudocode. A public partition of $[n]$ into $d$ groups, denoted $\pi$, is drawn uniformly.

An important subroutine is described by \eqref{eq:esti-encode}. It samples from $\pm 1$ in such a way that the mean is equal to the $j$th coordinate of user data $x_i$.

\begin{equation}
\label{eq:esti-encode}
\mathit{Encode}_\infty(x_i,j) := \begin{cases}
+1 & \mathrm{with~probability}~\half + \frac{x_{i,j}}{2}\\
-1 & \mathrm{with~probability}~\half - \frac{x_{i,j}}{2}
\end{cases}
\end{equation}

\begin{algorithm}
\caption{$\estir_{n,d,\eps,i}(x_i, \pi)$}
\label{alg:rir}

\Parameters{$n,d\in\Z^+,$ $\eps>0,$ $i\in [n]$}

\KwIn{$x_i \in B^d_\infty$; $\pi$, a public partition of $[n]$ into $d$ groups}
\KwOut{$y_i \in \{-\frac{e^\eps + 1}{e^\eps - 1}, \frac{e^\eps + 1}{e^\eps - 1} \}$}

\medskip

$g(i) \gets $ the group $i$ belongs to in $\pi$

$x'_i \gets \mathit{Encode}_\infty(x_i,g(i))$

$y_i \gets \rrr(x'_i)$

\Return{$y_i$}
\end{algorithm}

\begin{algorithm}
\caption{$\estia_{n,d,\eps}(y_1,\dots,y_n,\pi)$}
\label{alg:ria}

\Parameters{$n,d\in\Z^+, \eps>0$}

\KwIn{$\vec{y} \in \{-\frac{e^\eps + 1}{e^\eps - 1}, \frac{e^\eps + 1}{e^\eps - 1} \}^n$; $\pi$, a public partition of $[n]$ into $d$ groups}
\KwOut{$\vec{z} \in \R^d$}

\medskip

\For{$g:1\to d$}{
    $\pi(g) \gets$ the $g$th group of $[n]$

    $z_g \gets \frac{d}{n}\sum_{i \in \pi(g)} y_i$
}

\Return{$\vec{z}$}

\end{algorithm}

The following statement is a version of Theorem \ref{thm:esti-under-attack} that allows for arbitrary failure probability.

\begin{thm}
\label{thm:esti-under-attack-beta}
For any $\beta \in (0,1)$, there is a constant $c$ such that, for any $\eps > 0$, any positive integers $m \leq n$, any $x_1,\dots,x_n \in B^d_\infty$, and any attacker $M$ \emph{oblivious to public randomness}:
\[
\pr{}{\norm{\mathrm{Manip}(\esti_{n,d,\eps}, \vec{x}, M) - \frac{1}{n}\sum_{i=1}^n x_i }_\infty < c \cdot \frac{e^\eps + 1}{e^\eps - 1} \cdot \left( \sqrt{\frac{d}{n} \log \frac{d}{\beta}} + \frac{m}{n} \right)} \geq 1-\beta
\]
\end{thm}

To prove the theorem, we bound the error introduced by each source of randomness. We first consider the difference between the underlying mean and the mean in a subsample. Hoeffding's inequality and a union bound yields the following claim:
\begin{clm}
\label{clm:esti-subsample-error}
Fix any $x_1,\dots,x_n \in B^d_\infty$. There is a constant $c$ such that, when $\pi$ is a uniformly random partition of $[n]$ into $d$ groups $\pi(1),\dots,\pi(d)$,
\[
\pr{}{\forall g \in [d]~ \left| \frac{1}{n}\sum_{i=1}^n x_{i,g} - \frac{d}{n}\sum_{i\in \pi(g)} x_{i,g} \right| < c\cdot \sqrt{\frac{d}{n} \ln \frac{d}{\beta}}} \geq 1-\beta
\]
\end{clm}

When we encode all $x_{i,g(i)}$ with $\mathit{Encode}_\infty$, we may use the Hoeffding inequality again to bound the error:
\begin{clm}
\label{clm:esti-encode-error}
Fix any $x_1,\dots,x_n \in B^d_\infty$ and any partition $\pi$ of $[n]$ into $d$ groups $\pi(1),\dots,\pi(d)$. Suppose, for every $g\in[d]$, we execute $x'_i \gets \mathit{Encode}_\infty(x_i,g)$ for each user $i \in \pi(g)$.
\[
\pr{}{\forall g \in [d]~ \left|\frac{d}{n} \sum_{i\in \pi(g)} x_{i,g} - \frac{d}{n} \sum_{i\in \pi(g)} x'_i \right| < c\cdot \sqrt{\frac{d}{n} \ln \frac{d}{\beta} }} \geq 1-\beta
\]
\end{clm}

If an attacker chooses the set of corrupt users $C$ independently of $\pi$, we use a Chernoff bound to bound the number of corruptions in any group:
\begin{clm}
\label{clm:esti-weak-attacker}
Fix any $n,d\in \Z^+$ and any set of corrupted users $C \subset [n]$ where $|C| = m$. There is a constant $c$ such that, when $\pi$ is a uniformly random partition of $[n]$ into $d$ groups $\pi(1),\dots,\pi(d)$,
\[
\pr{}{\forall g \in [d] ~ |C\cap \pi(g)| < \frac{m}{d} + c\cdot \sqrt{\frac{m}{d} \ln\frac{d}{\beta}} } \geq 1-\beta
\]
\end{clm}

When we apply randomized response to data encoded by $\mathit{Encode}_\infty$, we can obtain our third bound on error immediately from Theorem \ref{thm:rr-under-attack}:
\begin{clm}
\label{clm:esti-privacy-error}
For any $m' \leq n/d$, any $\vec{x}\,' \in \{\pm 1 \}^{n/d}$ and any attacker $M$, $\rr_{\eps/2,n/d}$ has the following guarantee on estimation error after playing the $(m',n/d)$-manipulation game:
\[
\pr{}{\left|\mathrm{Manip}_{m',n/d}(\rr_{\eps, n/d}, \vec{x}', M) - \frac{d}{n}\sum_{i=1}^{n/d} x'_i \right| < \frac{e^\eps + 1}{e^\eps - 1} \cdot \left( \sqrt{\frac{2d}{n} \ln \frac{2d}{\beta}} + \frac{2dm'}{n} \right) } \geq 1-\beta/d
\]
\end{clm}

Theorem \ref{thm:esti-under-attack} follows from union bounds over Claims \ref{clm:esti-subsample-error}, \ref{clm:esti-encode-error}, \ref{clm:esti-weak-attacker}, and \ref{clm:esti-privacy-error} (we substitute $m'$ in Claim \ref{clm:esti-privacy-error} with the upper bound in Claim \ref{clm:esti-weak-attacker}).

\subsection{Construction and Analysis of \esto}
\label{apdx:esto}

The protocol $\esto_{n,d,\eps}$ consists of the $n$ randomizers $(\estor_{n,d, \eps, i})_{i\in[n]}$ and the aggregator $\estoa_{n,d,\eps}$ (see Algorithms \ref{alg:estor} and \ref{alg:estoa}, respectively). The vector for each user $i$, $\vec{s}_i \in \{\pm 1\}^{2d+1}$, is sampled uniformly and independently.

\begin{algorithm}
\caption{$\mathit{Encode}^d_1(x)$}

\Parameters{$d \in \Z^+$}

\KwIn{$x \in B^d_1$}
\KwOut{$x' \in [2d+1]$}

\medskip

\For{$j \in [d]$}{
    
    \If{$x_j > 0$}{
        $p_{2j-1} \gets x_j$
        
        $p_{2j} \gets 0$
    }
    \Else{
        $p_{2j-1} \gets 0$
    
        $p_{2j} \gets -x_j$
    }
}

$p_{2d+1} \gets 1 - \norm{x}_1$

Sample $x'$ from the distribution over $[2d+1]$ such that $\pr{}{x' = k} = p_k$

\Return{$x'$}
\end{algorithm}

\begin{algorithm}
\caption{$\estor_{n,d,\eps,i}(x_i, \vec{s}_i)$}
\label{alg:estor}

\Parameters{$n,d \in \Z^+, \eps > 0, i \in [n]$}

\KwIn{$x \in B^d_1$ and $\vec{s}_i \in \{\pm 1\}^{2d+1}$}
\KwOut{$y \in \{\pm \frac{e^\eps+1}{e^\eps-1} \}$}

\medskip

$x'_i \gets \mathit{Encode}^d_1(x_i)$

$y_i \gets \rrr_\eps(s_{i,x'})$

\Return{$y_i$}
\end{algorithm}

\begin{algorithm}
\caption{$\estoa_{n,d}(y_1,\vec{s}_1,\dots,y_n,\vec{s}_n)$}
\label{alg:estoa}

\Parameters{$n,d \in \Z^+$}

\KwIn{$y_i \in \{\pm \frac{e^\eps+1}{e^\eps-1} \}$ and $\vec{s}_i \in \{\pm 1\}^{2d+1}$ for every $i\in[n]$}
\KwOut{$\hat{\mu} \in \R^d$}

\medskip

\For{$j' \in [2d+1]$}{
    $z_{j'} \gets \frac{1}{n} \sum_{i=1}^n y_i s_{i,j'}$
}

\For{$j \in [d]$}{
    $\hat{\mu}_j \gets z_{2j-1} - z_{2j}$
}

$\hat{\mu} \gets (\hat{\mu}_1, \dots, \hat{\mu}_d)$

\Return{$\hat{\mu}$}
\end{algorithm}

\subsubsection{Error in $\ell_\infty$}

\begin{thm}
\label{thm:esto-infty}
There is a constant $c$ such that, for any $\beta \in (0,1)$, any $\eps > 0$, any positive integers $n,d$, and any $\vec{x}=(x_1,\dots,x_n) \in [d]^n$, with probability $\geq 1-\beta$, we have
\[
\norm{\esto_{n,d,\eps}(\vec{x}) - \frac{1}{n}\sum_{i=1}^n x_i }_\infty \leq c\cdot \frac{e^\eps+1}{e^\eps - 1} \cdot  \sqrt{\frac{1}{n} \ln \frac{d}{\beta}}
\]
\end{thm}

To prove the theorem, we bound the error introduced by $\mathit{Encode}^d_\infty$ and by $\rrr_\eps$ separately. Recall the shorthand $\mathit{freq}(j,\vec{x}) := \frac{1}{n} \sum_{i=1}^n \indic{x_i = j}$ and $\mathit{freq}(\vec{x}):= (\mathit{freq}(1,\vec{x}), \dots, \mathit{freq}(d,\vec{x}))$.

\begin{clm}
There is a constant $c$ such that for any positive integers $n,d$ and any $x_1,\dots,x_n\in B^d_1$, if we sample $x'_i\gets \mathit{Encode}^d_1(x_i)$ for each user $i$, then
\[
\pr{}{\max_{j \in [d]}\left| (\mathit{freq}(2j-1,\vec{x}\,') - \mathit{freq}(2j,\vec{x}\,')) - \frac{1}{n}\sum_{i=1}^n x_{i,j} \right| \leq c\cdot \sqrt{\frac{1}{n}\ln \frac{d}{\beta} } } \geq 1-\beta
\]
\end{clm}

\begin{proof}
Consider any user data $x_i \in B^d_1$ and any coordinate $j\in [d]$. Without loss of generality, we will assume that $x_{i,j} > 0$. By construction, $\pr{}{x'_i=2j-1} = x_i$ and $\pr{}{x'_i = 2j}=0$. Hence,
\begin{align*}
\ex{}{\indic{x'_i=2j-1}-\indic{x'_i=2j}} &= x_{i,j}\\
\ex{}{\sum_{i=1}^n\indic{x'_i=2j-1}-\sum_{i=1}^n\indic{x'_i=2j}} &= \sum_{i=1}^n x_{i,j}\\
\ex{}{\mathit{freq}(2j-1,\vec{x}\,') - \mathit{freq}(2j,\vec{x}\,')}&= \frac{1}{n}\sum_{i=1}^n x_{i,j}
\end{align*}
The random variable $\indic{x'_i=2j-1}-\indic{x'_i=2j}-x_{i,j}$ ranges from $-2$ to $+2$. By a Hoeffding bound, the following holds with probability $\geq 1-\beta/d$.
\[
\left| \mathit{freq}(2j-1,\vec{x}\,') - \mathit{freq}(2j,\vec{x}\,') - \frac{1}{n}\sum_{i=1}^n x_{i,j} \right| < \sqrt{\frac{8}{n}\ln \frac{2d}{\beta} }
\]
A union bound over all $j\in[d]$ completes the proof.
\end{proof}

\begin{clm}
There is a constant $c$ such that for any positive integers $n,d$ and any $x'_1,\dots,x'_n \in [2d+1]$, if we sample $y_i \gets \rrr_\eps(s_{i,x'_i})$ for each user $i$ and compute $z_{j'} = \frac{1}{n} \sum_{i=1}^n y_i s_{i,j'}$ for each $j'\in[2d+1]$, then
\[
\pr{}{\norm{\vec{z} - \mathit{freq}(\vec{x}\,')}_\infty \leq c \cdot \frac{e^\eps+1}{e^\eps - 1} \cdot \sqrt{\frac{1}{n} \ln \frac{d}{\beta}}} \geq 1-\beta
\]
\end{clm}

\begin{proof}
To prove this claim, we fix a value $j'\in[2d+1]$ and argue that the estimate of $\mathit{freq}(j',\vec{x}\,')$ has error $c \cdot \frac{e^\eps+1}{e^\eps - 1} \cdot \sqrt{\frac{1}{n} \ln \frac{d}{\beta}}$ with probability $1-\beta/(2d+1)$. A union bound over all $j'\in [2d+1]$ will complete the proof.

Recall that all $y_i$ have magnitude $\frac{e^\eps+1}{e^\eps - 1}$. By a Hoeffding bound, the following holds with probability $\geq 1-\beta/(2d+1)$ (where we use $\vec{s}(j')$ to denote the vector $(s_{1,j'}, \dots, s_{n,j'})$:
\[
\left| z_{j'} - \ex{\rrr,\vec{s}(j')}{z_{j'}} \right| < \frac{e^\eps+1}{e^\eps - 1} \cdot \sqrt{\frac{2}{n} \ln \frac{2(2d+1)}{\beta}}
\]
It remains to show that $z_{j'}$ is unbiased:
\begin{align*}
\ex{\rrr,\vec{s}(j')}{z_{j'}} &= \frac{1}{n} \sum_{i=1}^n \ex{\rrr, \vec{s}(j')}{y_i s_{i,j'}}\\
    &= \frac{1}{n} \sum_{i=1}^n \ex{\vec{s}(j')}{s_{i, j'} \cdot  s_{i,x_i}}\\
    &= \frac{1}{n} \sum_{i=1}^n \indic{x_i = j'} = \mathit{freq}(j', \vec{x}\,')
\end{align*}
This concludes the proof.
\end{proof}

We also formalize the argument that $m$ users cannot manipulate the protocol beyond $O(m / \eps n)$:
\begin{thm}
\label{thm:esto-manip-infty}
There is a constant $c$ such that, for any $\beta \in (0,1)$, any $\eps > 0$, any positive integers $m,n,d$, any attacker $M$, and any $\vec{x}=(x_1,\dots,x_n) \in [d]^n$, with probability $\geq 1-\beta$, we have
\[
\norm{\mathrm{Manip}_{m,n}(\esto_{n,d,\eps},\vec{x},M) - \frac{1}{n}\sum_{i=1}^n x_i }_\infty \leq c\cdot \frac{e^\eps+1}{e^\eps - 1} \cdot \left( \sqrt{\frac{1}{n} \ln \frac{d}{\beta}} + \frac{m}{n} \right)
\]
\end{thm}
\begin{proof}
As sketched in Section \ref{sec:protocols}, we bound the error from the honest execution separately from the error from the manipulation. Given that Theorem \ref{thm:esto-infty} already bounds the honest execution, it will suffice to prove that
\[
\max_{j \in [2d+1]} \left| \frac{1}{n} \sum_{i\in C} (y_i - \underline{y}_i) s_{i,j'} \right| \leq 2 \cdot \frac{e^\eps + 1}{e^\eps -1} \cdot \frac{m}{n}
\]
By construction, both the manipulative messages $y_i$ and the honest messages $\underline{y}_i$ are members of the set $\{\pm \frac{e^\eps + 1}{e^\eps -1}\}$. There are $m$ corrupt users in $C$. Hence, the bound follows.
\end{proof}

\subsubsection{Error in $\ell_1$}
\begin{thm}
\label{thm:esto-l1}
There is a constant $c$ such that, for any $\beta \in (0,1)$, any $\eps > 0$, any positive integers $n,d$, and any $\vec{x}=(x_1,\dots,x_n) \in [d]^n$, with probability $\geq 1-\beta$, we have
\[
\norm{\esto_{n,d,\eps}(\vec{x}) - \frac{1}{n}\sum_{i=1}^n x_i }_1 \leq c\cdot \frac{e^\eps+1}{e^\eps - 1} \cdot \sqrt{\frac{d^2}{n} \log\frac{1}{\beta}}
\]
\end{thm}

To prove the theorem, we bound the error introduced by $\mathit{Encode}^d_\infty$ and by $\rrr_\eps$ separately.

\begin{clm}
\label{clm:esto-l1-encoding}
There is a constant $c$ such that for any positive integers $n,d$ and any $x_1,\dots,x_n\in B^d_1$, if we sample $x'_i\gets \mathit{Encode}^d_1(x_i)$ for each user $i$, then the following holds with probability $\geq 1-\beta$:
\begin{equation}
\label{eq:esto-l1-encoding}
\sum_{j \in [d]}\left| (\mathit{freq}(2j-1,\vec{x}') - \mathit{freq}(2j,\vec{x}')) - \frac{1}{n} \sum_{i=1}^n x_{i,j} \right| \leq c \cdot \sqrt{\frac{d^2}{n} \log\frac{1}{\beta}}
\end{equation}
\end{clm}

\begin{proof}
For any $i\in[n], j\in[d]$, define the random variable $\mathit{err}(i,j):=\indic{x'_i=2j-1}-\indic{x'_i=2j} - x_{i,j}$. Observe that $\ex{}{\mathit{err}(i,j)} = 0$ and $|\mathit{err}(i,j)|\leq 2$. By Hoeffding's inequality, the quantity $\frac{1}{n} \sum_{i=1}^n \mathit{err}(i,j)$ is subgaussian. Specifically, for all $t > 0$,
\begin{equation}
\label{eq:esto-l1-encoding-subg}
\pr{}{\left|\frac{1}{n} \sum_{i=1}^n \mathit{err}(i,j)\right| > t } \leq 2\exp(-nt^2/8)
\end{equation}
Note that this implies there are constants $c_0,c_1$ such that
\begin{align}
\ex{}{\left|\frac{1}{n} \sum_{i=1}^n \mathit{err}(i,j)\right|} &\leq c_0 \cdot \sqrt{\frac{1}{n}} \label{eq:esto-l1-encoding-ex} \\
\var{}{\left|\frac{1}{n} \sum_{i=1}^n \mathit{err}(i,j)\right|} &\leq c_1 \cdot \frac{1}{n} \label{eq:esto-l1-encoding-var}
\end{align}

For shorthand, we define $\mathit{err}(j) := \left|\frac{1}{n} \sum_{i=1}^n \mathit{err}(i,j)\right| \leq 2$. Observe that the left-hand side of \eqref{eq:esto-l1-encoding} is equivalent to $\sum_{j=1}^d \mathit{err}(j)$. From \eqref{eq:esto-l1-encoding-ex}, \eqref{eq:esto-l1-encoding-ex}, and a Chernoff bound, the sum tightly concentrated around its expectation:
\begin{align*}
\pr{}{\sum_{j=1}^d \mathit{err}(j) > d\cdot \ex{}{\mathit{err}(1)} + \sqrt{d \var{}{\mathit{err}(1)} \log \frac{1}{\beta}} } &\leq \beta\\
\pr{}{\sum_{j=1}^d \mathit{err}(j) > c_0\cdot \sqrt{\frac{d^2}{n}} + \sqrt{c_1\cdot \frac{d}{n} \log \frac{1}{\beta}} } &\leq \beta \tag{From \eqref{eq:esto-l1-encoding-ex} and \eqref{eq:esto-l1-encoding-var}}
\end{align*}
This concludes the proof.
\end{proof}

\begin{clm}
There is a constant $c$ such that for any positive integers $n,d$ and any $x'_1,\dots,x'_n \in [2d+1]$, if we sample $y_i \gets \rrr_\eps(s_{i,x'_i})$ for each user $i$ and compute $z_{j'} = \frac{1}{n} \sum_{i=1}^n y_i s_{i,j'}$ for each $j'\in[2d+1]$, then the following holds with probability $\geq 1-\beta$
\[
\sum_{j' \in [2d+1]} \left| z_{j'} - \mathit{freq}(j',\vec{x}') \right|\leq c \cdot \frac{e^\eps+1}{e^\eps - 1} \cdot \sqrt{\frac{d^2}{n} \log \frac{1}{\beta}}
\]
\end{clm}
\begin{proof}
Define the random variable $\mathit{err}(i,j') := y_i s_{i,j'} - \indic{x'_i=j'}$. The same steps taken in the proof of Claim \ref{clm:esto-l1-encoding} apply here, except now $|\mathit{err}(i,j')| \leq 2\cdot \frac{e^\eps+1}{e^\eps-1}$.
\end{proof}

\subsection{Construction and Analysis of \raptor}
\label{apdx:raptor}

The protocol $\raptor_{n,G,\eps}$ consists of $G$ randomizers: user $i$ is assigned randomizer $\lceil i / G \rceil$. Public randomness will generate $S_1,\dots,S_G$ each of which are uniformly random subsets of $[d]$ of size $d/2$. If a user runs the $g$-th randomizer, they will privately report whether or not their data lies in $S_g$. The aggregator performs a threshold test on each group.

\newcommand{\pmindic}[1]{\mathbbm{1}_{\pm} \left[#1\right]}
We reproduce the randomizer and aggregator pseudocode in Algorithms \ref{alg:raptorr} and \ref{alg:raptora}. For the sake of this proof, we use $\pmindic{\mathit{bool}}$ to denote the indicator function that evaluates to $+1$ when $\mathit{bool}$ is true and $-1$ when it false.

\begin{algorithm}
\caption{$\raptorr_{\eps, g}$, randomizer for uniformity testing}
\label{alg:raptorr}

\Parameters{Privacy parameter $\eps$; group number $g \in [G]$}

\KwIn{$x \in [d]$; public sets $S_1,\dots,S_G$ where $S_g \subset [d]$ and $|S_g|=d/2$}
\KwOut{$y \in \{\pm \frac{e^\eps +1}{e^\eps-1}\}$}

\medskip

$x' \gets \pmindic{x \in S_g}$

$y \sim \rrr_\eps(x')$

\Return{$y$}
\end{algorithm}

\begin{algorithm}
\caption{$\raptora_{n,G,\eps}$, an aggregation algorithm for uniformity testing}
\label{alg:raptora}

\Parameters{Positive integers $n,G$; privacy parameter $\eps$}

\KwIn{$y_1, \dots y_n \in \{\pm \frac{e^\eps +1}{e^\eps-1}\}^n$; public sets $S_1,\dots,S_G$}
\KwOut{The string ``Uniform'' or the string ``Not Uniform''}

$\alpha_G := \frac{e^\eps + 1}{e^\eps - 1} \cdot \left( \sqrt{\frac{6G}{n} \ln \frac{4G}{\beta}} + \frac{2mG}{n} \right)$

\medskip

\For{$g \in [G]$}{

    $\mathit{start}(g) \gets 1 + (g-1)\cdot n / G$
    
    $\mathit{end}(g) \gets g\cdot n /G$

    \tcc{Estimate of probability mass in $S_g$}
    $\tilde{p}(S_g) \gets \frac{G}{n}\sum_{i=\mathit{start}(g)}^{\mathit{end}(g)} y_i$

    \If{$|\tilde{p}(S_g)| > 2\alpha_G$}{
        \Return{``Not uniform''}
    }
}

\Return{``Uniform''}
\end{algorithm}

We rely on the following technical lemma concerning uniformly random $S$:
\begin{lem}[From \cite{AcharyaCFT19}]
\label{lem:acft-concentration}
If $S$ is a uniformly random subset of $[d]$ with size $d/2$ and $\norm{\bP - \bU}_1 > \alpha \sqrt{10 d}$, then
\[
\pr{S}{\big|\half - \pr{x \sim \bP}{x\in S} \big| > \alpha} > \frac{1}{477}
\]
\end{lem}

\begin{coro}
\label{coro:acft-concentration}
If $S$ is a uniformly random subset of $[d]$ with size $d/2$ and $\norm{\bP - \bU}_1 > \alpha \sqrt{10 d}$, then
\[
\pr{S}{\big|\ex{x\sim \bP}{\pmindic{x\in S}} \big| > 2\alpha} > \frac{1}{477}
\]
\end{coro}

The following statement is a version of Theorem \ref{thm:raptor} that allows for arbitrary failure probability $\beta$.

\begin{thm*}
There is a constant $c$ and a choice of parameter $G = \Theta(\log 1/\beta)$ such that, for any $\eps>0$, any positive integers $m\leq n$, and any attacker $M$, the following holds with probability $\geq 1-\beta$
\[
\mathrm{Manip}_{m,n}(\raptor_{n,G,\eps}, \bU, M) = \mathrm{``uniform"}
\]
and, when $\norm{\bP-\bU}_1 \geq c\cdot \frac{e^\eps + 1}{e^\eps - 1} \cdot \left(\sqrt{\frac{dG}{n} \ln \frac{G}{\beta}} + \frac{m G\sqrt{d}}{n} \right)$, the following also holds with probability $\geq 1-\beta$
\[
\mathrm{Manip}_{m,n}(\raptor_{n,G,\eps}, \bP, M) = \mathrm{``not~uniform"}
\]
\end{thm*}

\begin{proof}
We specify the following undesirable events:
\begin{align*}
    E_1 &:= \exists g \in [G] ~ \big|\ex{x\sim \bP}{\pmindic{x\in S}} - \frac{G}{n} \sum_{i=\mathit{start}(g)}^{\mathit{end}(g)} \pmindic{x_i \in S} \big| > \alpha_G\\
    E_2 &:= \exists g \in [G] ~ \big|\frac{G}{n} \sum \pmindic{x_i \in S} - \tilde{p}(S_g) \big| > \alpha_G\\
    E_2 &:= \forall g \in [G]~ \big| \ex{x\sim \bP}{\pmindic{x\in S}} \big| < 2\alpha_{ G}
\end{align*}

If $\bP=\bU$ and neither $E_1$ nor $E_2$ have occurred, every $\tilde{p}(S_{g})$ is at most $2\alpha_G$. Thus, the output is ``Uniform.''

If $\norm{\bP-\bU}_1 \geq \alpha_G\cdot \sqrt{160d}$ and none of $E_1,E_2,E_3$ have occurred, some $\tilde{p}(S_{g})$ has magnitude at least $2\alpha_{G}$. Thus, the output is ``Not uniform.''

$\pr{}{E_1} < \beta / 3$ follows from a Hoeffding bound. $\pr{}{E_2} < \beta / 3$ follows from Theorem \ref{thm:rr-under-attack} and a union bound over $G$ plays of $\rr_{n/G,\eps}$ in the manipulation game. When $\norm{\bP-\bU}_1 \geq \sqrt{160d}\cdot \alpha_G$ and $G \gets \lceil \ln(2/\beta) / \ln(477/476) \rceil$, $\pr{}{E_3} < \beta / 3$ follows from Corollary \ref{coro:acft-concentration}. A union bound over all three completes the proof.
\end{proof}

\subsection{Construction and Analysis of \hh}
\label{apdx:hh}

The protocol $\hh_{n,d,k,\eps}$ consists of the $n$ randomizers $(\hhr_{n,d,k,\eps, i})_{i\in[n]}$ and the aggregator $\hha_{n,d,k,\eps}$; see Algorithms \ref{alg:hhr} and \ref{alg:hha} for the pseudocode. A public data structure $\pi$ partitions $[n]$ into $\log_2 d$ groups uniformly at random. We assume the data structure has an implicit order within each group $\pi(1),\dots,\pi(\log_2 d)$. The public hash function $h:[d]\to[k]$ is drawn uniformly. To facilitate the use of $\esto$, we also sample vectors $\vec{s}_1,\dots,\vec{s}_n$ uniformly from $\{\pm 1\}^{2k}$.

\begin{algorithm}
\caption{$\texttt{OneHotHash}_{h,k}(g,x_i)$}
\label{alg:onehothash}

\tcc{One-hot vector that encodes a hashed value}

$x'_i \gets (\underbrace{0,\dots,0}_{2k~\mathrm{copies}})$

$h_g(x) \gets 2h(x) - \mathit{bit}(g,x_i)$

$x'_{i,h_g(x)} \gets +1$

\Return{$x'_i$}
\end{algorithm}

\begin{algorithm}
\caption{$\hhr_{n,d,k,\eps,i}(x_i,\pi,h,\vec{s}_i)$}
\label{alg:hhr}

\Parameters{$n,d,k\in\Z^+$; $\eps > 0$; $i\in [n]$}

\KwIn{$x_i \in [d]$; public partition $\pi$; public hash $h:[d]\to[k]$; public vector $\vec{s} \in \{\pm 1\}^{2k}$}
\KwOut{$y_i \in \{\pm \frac{e^\eps + 1}{e^\eps - 1}\}$}

\medskip

$g(i) \gets $ group that $i$ belongs to in $\pi$

$x'_i \gets \texttt{OneHotHash}_{h,k}(g(i),x_i)$

\tcc{Contribute to a histogram by reporting the one-hot}
$n'\gets n / \log_2 d$

$i' \gets $ index of $i$ in group $g(i)$

$y_i \sim \estor_{n',2k,\eps,i'}(x'_i, \vec{s}_i)$

\Return{$y_i$}
\end{algorithm}

\begin{algorithm}
\caption{$\hha_{n,d,k,\eps}(y_1,\vec{s}_1,\dots,y_n,\vec{s}_n,\pi,h)$}
\label{alg:hha}

\Parameters{$n,d,k\in\Z^+$; $\eps > 0$}

\KwIn{$y_i \in \{\pm \frac{e^\eps + 1}{e^\eps - 1}\}$ and  $\vec{s}_i \in \{\pm 1\}^{ 2k}$ for each $i\in[n]$; public partition $\pi$; public hash $h:[d]\to[k]$}
\KwOut{$L \subset [d]$ with size $k$}

\medskip

$n'\gets n / \log_2 d$

\For{$g \from 1 \to \log_2 d$}{
    \tcc{Obtain the $g$-th noisy histogram}
    
    $(z_1^{(g)},\dots,z_{2k}^{(g)}) \gets \estoa_{n',2k,\eps}( \{y_i, \vec{s}_i\}_{i \in \pi(g)})$
}

\medskip

$L \gets \emptyset$

\For{$v \in [k]$}{
    \tcc{Recover the bits of an element in $[d]$ that hashes to $v$}
    
    \For{$g \from 1 \to \log_2 d$}{
        \lIf{$z_{2v-1}^{(g)} > z_{2v}^{(g)}$}
            { $\mathit{bit}_v^{(g)} \gets 1$ }
        \lElse{ $\mathit{bit}_v^{(g)} \gets 0$ }
    }
    
    $L \gets L~ \cup$ the number represented in binary by $\mathit{bit}_v^{(1)}, \dots, \mathit{bit}_v^{(\log_2 d)}$
}
\Return{$L$}
\end{algorithm}

\begin{thm}
\label{thm:hh-under-attack-formal}
There is a constant $c$ such that, for any $\eps >0$, any positive integers $m\leq n$, any $\vec{x}=(x_1,\dots,x_n) \in [d]^n$, and any adversary $M$, if we execute $L \gets \mathrm{Manip}_{m,n}(\hh_{n,k,\eps},\vec{x},M)$ with parameter $k\gets 3n^2/\beta$, then with probability $\geq 1-\beta$, $L$ contains all $j$ such that
\[
\mathit{freq}(j,\vec{x}) > c\cdot \frac{e^\eps+1}{e^\eps - 1} \cdot \left(\sqrt{\frac{\log d}{n}  \log \frac{n \log d}{\beta}} \right) + \frac{m\log d}{n}
\]
\end{thm}

As discussed in Section \ref{sec:hh-protocol}, there are three undesirable events that can occur when the game is played. in three separate claims below, we state them formally and bound the probability of each event by $\beta / 3$. We first consider the event that the frequency of any $j \in \vec{x}$ is significantly different from the frequency of $j \in \vec{x}^{(g)}$: 

\begin{clm}
\label{clm:hh-subsample-error}
Fix any $\vec{x} \in [d]^n$. There is a constant $c$ such that, when $\pi$ is a uniformly random partition of $[n]$ into groups $\pi(1),\dots,\pi(\log_2 d)$ each of size $n/\log_2 d$,
\[
\pr{}{\forall j \in \vec{x} ~ \forall g \in [\log_2 d]~ |\mathit{freq}(j,\vec{x}) - \mathit{freq}(j,\vec{x}^{(g)})| > c\cdot \sqrt{\frac{\log_2 d}{n} \ln \frac{n\cdot \log_2 d}{\beta}}} \leq \beta/3
\]
\end{clm}

This is proven via a Hoeffding bound and a union bound. Next we argue that there are likely no collisions:

\begin{clm}
\label{clm:hh-collision}
If $k > 3n^2 / \beta$, then for any $\vec{x} \in [d]^n$ and a uniformly chosen $h:[d]\to[k]$,
\[
\pr{}{\exists x\neq x'\in\vec{x} ~~~ h(x) = h(x')} \leq \beta/3
\]
\end{clm}
\begin{proof}
The argument is brief:
\begin{align*}
\pr{}{\exists x\neq x'\in\vec{x} ~~~ h(x) = h(x')} &\leq n\cdot \pr{}{\exists x\neq x_1 ~~~ h(x) = h(x_1)}\\
    &\leq n^2\cdot \pr{}{h(x_2)=h(x_1)}\\
    &= \frac{n^2}{k} < \beta/3
\end{align*}
\end{proof}

A core part of the protocol is, for each group $g$, the execution of $\esto_{n',2k,\eps}$ on one-hot encodings $(x'_i)_{i\in \pi(g)}$. Theorem \ref{thm:esto-manip-infty} implies the following:

\begin{clm}
\label{clm:hh-privacy-error}
Fix any $m < n'$, $\vec{x}~' \in (B^{2k}_1)^{n'}$, any adversary $M$ against $\esto_{n',2k,\eps}$, and any $\beta\in(0,1)$. There exists a constant $c$ such that
\[
\pr{}{\norm{\mathrm{Manip}_{m,n'}(\esto_{n',2k,\eps},\vec{x}~'',M) - \frac{1}{n'}\sum_{i=1}^{n'} x'_i}_\infty > c\cdot \frac{e^\eps+1}{e^\eps - 1} \cdot \left(\sqrt{\frac{1}{n'}  \ln \frac{k}{\beta}} \right) + \frac{m}{n'}} \leq \beta/3
\]
\end{clm}

We are now ready to prove Theorem \ref{thm:hh-under-attack-formal}

\begin{proof}[Proof of Theorem \ref{thm:hh-under-attack-formal}]
Let $c_0,c_1$ be the constants from Claims \ref{clm:hh-subsample-error} and \ref{clm:hh-privacy-error}, respectively. We will prove that with probability $\geq 1-\beta$, for each $g \in [\log_2 d]$ and for each $j$ such that
\begin{equation}
\label{eq:hh-threshold}
\mathit{freq}(j,\vec{x}) > c_0\cdot \sqrt{\frac{\log_2 d}{n} \ln \frac{n\cdot \log_2 d}{\beta}} + 2c_1 \cdot \frac{e^\eps+1}{e^\eps - 1} \cdot \left( \sqrt{\frac{1}{n'}  \ln \frac{k \log_2 d}{\beta}} + \frac{m}{n'} \right)
\end{equation}
the protocol will reconstruct the $g$-th bit of $j$.

Each user constructs $x'_i\gets \texttt{OneHotHash}_{h,k}(g(i), x_i)$ when running $\hhr$ on their data. A union bound over Claims \ref{clm:hh-subsample-error} and \ref{clm:hh-collision} implies the following two inequalities hold for all groups $g$ (with probability $1-2\beta/3$):
\begin{align}
\frac{1}{n'}\sum_{i\in \pi(g)} x'_{i, 2\cdot h(j)-\mathit{bit}(g,j)} &\geq \mathit{freq}(j,\vec{x})-c_0\cdot \sqrt{\frac{\log_2 d}{n} \ln \frac{n\cdot \log_2 d}{\beta}} \label{eq:hh-freq-high} \\
\frac{1}{n'}\sum_{i\in \pi(g)} x'_{i,2\cdot h(j)+\mathit{bit}(g,j)-1} &=0 \label{eq:hh-freq-zero}
\end{align}

From Claim \ref{clm:hh-privacy-error}, the following two inequalities hold for all $g$ (with probability $1-\beta/3$):
\begin{align}
|z^{(g)}_{2\cdot h(j)-\mathit{bit}(g,j)} - \frac{1}{n'}\sum_{i\in \pi(g)} x'_{i, 2\cdot h(j)-\mathit{bit}(g,j)}| &\leq c_1 \cdot \frac{e^\eps+1}{e^\eps - 1} \cdot \left( \sqrt{\frac{1}{n'}  \ln \frac{k \log_2 d}{\beta}} + \frac{m}{n'}\right) \label{eq:hh-acc-high}\\
|z^{(g)}_{2\cdot h(j)+\mathit{bit}(g,j)-1} - \frac{1}{n'}\sum_{i\in \pi(g)} x'_{i,2\cdot h(j)+\mathit{bit}(g,j)-1}| &\leq c_1 \cdot \frac{e^\eps+1}{e^\eps - 1} \cdot \left( \sqrt{\frac{1}{n'}  \ln \frac{k \log_2 d}{\beta}} + \frac{m}{n'}\right) \label{eq:hh-acc-zero}
\end{align}

By a union bound, the following holds with probability $\geq 1-\beta$:
\begin{align*}
z^{(g)}_{2\cdot h(j)-\mathit{bit}(g,j)} &\geq \frac{1}{n'} \sum_{i\in \pi(g)} x'_{i, 2\cdot h(j)-\mathit{bit}(g,j)} - c_1 \cdot \frac{e^\eps+1}{e^\eps - 1} \cdot \left( \sqrt{\frac{1}{n'}  \ln \frac{k \log_2 d}{\beta}} + \frac{m}{n'}\right) \tag{From \eqref{eq:hh-acc-high}}\\
    &\geq \mathit{freq}(j,\vec{x})-c_0\cdot \sqrt{\frac{\log_2 d}{n} \ln \frac{n\cdot \log_2 d}{\beta}} - c_1 \cdot \frac{e^\eps+1}{e^\eps - 1} \cdot \left( \sqrt{\frac{1}{n'}  \ln \frac{k \log_2 d}{\beta}} + \frac{m}{n'}\right) \tag{From \eqref{eq:hh-freq-high}}\\
    &> c_1 \cdot \frac{e^\eps+1}{e^\eps - 1} \cdot \left( \sqrt{\frac{1}{n'}  \ln \frac{k \log_2 d}{\beta}} + \frac{m}{n'}\right) \tag{From \eqref{eq:hh-threshold}}\\
    &\geq z^{(g)}_{2\cdot h(j)+\mathit{bit}(g,j)-1} + \frac{1}{n'} \sum_{i\in \pi(g)} x'_{i,2\cdot h(j)+\mathit{bit}(g,j)-1} \tag{From \eqref{eq:hh-acc-zero}}\\
    &= z^{(g)}_{2\cdot h(j)+\mathit{bit}(g,j)-1} \tag{From \eqref{eq:hh-freq-zero}}
\end{align*}
By construction, this means we will assign $\mathit{bit}^{(g)}_{h(j)} \gets \mathit{bit}(g,j)$. this concludes the proof.
\end{proof}

\fi

\end{document}